\documentclass[11pt]{article}
\usepackage{color,soul}
\usepackage{titlesec}
\usepackage{hyperref}
\usepackage{dirtytalk}
\usepackage{lipsum}
\usepackage{ntheorem}
\usepackage{mathtools}
\usepackage{graphicx}
\usepackage{calc}
\newlength{\depthofsumsign}
\usepackage{graphicx}
\usepackage{amsmath}
\usepackage{amssymb}
\usepackage{amsfonts}
\usepackage{array}
\usepackage{tikz}
\usepackage{hhline}
\usepackage{multirow}
\usepackage{subfig}
\usepackage{hyperref}
\usepackage{mathtools}
\usepackage{tikz}
\usepackage{makecell}
\usepackage{pgfplots}
\usepackage{afterpage}
\usepackage{fancyhdr}
\usepackage{setspace}
\usepackage{algorithmic}
\usepackage{bm}

\newtheorem{corollary}{Corollary}
\newtheorem{proposition} {Proposition}
\newtheorem*{proposition-non}{Proposition}

\newtheorem{theorem}{Theorem}

\newtheorem{remark}{Remark}

\allowdisplaybreaks

\newenvironment{proof}{\noindent{\bf Proof:}\indent}%
                      {\hfill $\Box$\par}

\newcommand{\sym}[1]{{\sf #1}}

\title{Influence of a Set of Variables on a Boolean Function}
\author{Aniruddha Biswas and Palash Sarkar \\
Indian Statistical Institute \\
203, B.T.Road, Kolkata \\
India 700108. \\
Email: \{anib\_r, palash\}@isical.ac.in
}
\date{\today}

\begin{document}

\maketitle

\begin{abstract}
	The influence of a variable is an important concept in the analysis of Boolean functions.  
	The more general notion of influence of a set of variables on a Boolean function has four separate definitions in the literature.
	In the present work, we introduce a new definition of influence of a set of variables which is based on the auto-correlation function and develop its basic
	theory. Among the new results that we obtain are generalisations of the Poincar\'{e} inequality and the edge expansion property of the influence of a single variable.
	Further, we obtain new characterisations of resilient and bent functions using the notion of influence. We show that the previous definition of influence
	due to Fischer et al. (2002) and Blais (2009) is half the value of the auto-correlation based influence that we introduce. Regarding the other
	prior notions of influence, we make a detailed study of these and show that each of these definitions do not satisfy one or more desirable properties
	that a notion of influence may be expected to satisfy.  
\end{abstract}

\section{Introduction \label{sec-intro} }
Boolean functions play an important role in diverse areas of mathematics and computer science, including combinatorics, probability, complexity theory, learning theory,
cryptography and coding theory. We refer to two excellent books on Boolean functions, namely~\cite{o2014analysis} and~\cite{carlet2021boolean}. The first book
focuses on Boolean functions in the context of theoretical computer science, while the second book focuses on Boolean functions in relation to cryptography and
coding theory. 

The notion of influence of a variable on a Boolean function was introduced by Ben-Or and Linial~\cite{ben1987collective}. Subsequently, this concept has become central
to the study of Boolean functions in various contexts. See~\cite{o2014analysis} for a very comprehensive account of such applications. The notion of influence, however,
has not received much attention in the context of cryptographic applications of Boolean functions. We know of only two works~\cite{DBLP:journals/iacr/GangopadhyayS14,BS2021} which 
studied influence in relation to cryptographic properties. 

The notion of influence of a variable on a function has been extended to consider the influence of a set of variables on a function. We have been able to locate four different
definitions of the influence of a set of variables on a Boolean function. The first definition appears in the work of Ben-Or and Linial~\cite{ben1987collective} itself in 1989.
A different definition due to Fischer et al.~\cite{fischer02} appeared in 2002 and the same definition was considered in 2009 by Blais~\cite{blais2009testing}.
A third definition was given by Gangopadhyay and St{\u{a}}nic{\u{a}}~\cite{DBLP:journals/iacr/GangopadhyayS14} in 2014 and a fourth definition was given by Tal~\cite{tal2017tight} in 2017.
All of these definitions coincide with each other in the case of a single variable, but in the case of more than one variable, in general the values provided by the 
four definitions of influence are different. 

The motivation of our work is to make a systematic and comprehensive study of the notion of influence of a set of variables on a Boolean function. To this end,
we introduce a definition of influence based on the auto-correlation function, which is a very useful tool for analysing certain cryptographic properties of Boolean functions.
Two Walsh transform based characterisations of influence are obtained and some basic intuitive properties are derived. Several results on the influence of a single variable
are generalised. These include Poincar\'{e} inequality and edge expansion property of influence of a variable. In the context of cryptographic properties, we provide characterisations
of resilient and bent functions using the notion of influence. 

The definition of influence given in~\cite{fischer02,blais2009testing} is shown to be half the value of
the notion of influence that we introduce. We also argue that the definition of influence considered in~\cite{DBLP:journals/iacr/GangopadhyayS14} does not satisfy a basic desirable
property, namely that the influence of a set of variables can be zero even if the function is not degenerate on these variables. 

Next we define a quantity called pseudo-influence, obtain its Walsh transform based characterisation and derive certain basic properties. We show that the pseudo-influence
does not satisfy some intuitive properties that one would expect a notion of influence to satisfy, which is why we call it pseudo-influence. From the Walsh transform
based characterisation, it follows that the definition of influence considered by Tal~\cite{tal2017tight} is the notion of pseudo-influence that we introduce. 
Our motivation for introducing pseudo-influence and analysing it is to show that the notion of influence considered in~\cite{tal2017tight} is not satisfactory.

Lastly, we make a systematic study of the Ben-Or and Linial (BL) notion of influence~\cite{ben1987collective}.
We show that the BL notion of influence satisfies some desirable properties, but it does not satisfy sub-additivity. Further, we argue that compared to the 
auto-correlation based definition, the BL notion of influence is a more coarse measure.


Section~\ref{sec-prelim} introduces the background and the notation and also describes the previous definitions of influence of a set of variables.
The definition of influence from auto-correlation is introduced in Section~\ref{sec-inf-set} and its Walsh transform based characterisations and basic
properties are derived. The concept is further developed in several subsections.
The path expansion property of influence is derived in Section~\ref{subsec-path}, two probabilistic interpretations of influence are
given in Section~\ref{subsec-prob}, the relation of influence to juntas and cryptographic properties are described in Section~\ref{subsec-junta} 
and~\ref{subsec-crypto} respectively, and a general form the Fourier entropy/influence conjecture is mentioned in Section~\ref{subsec-FEI}. The notion of pseudo-influence 
is defined in Section~\ref{sec-PI} and its properties as well as its relation to influence are studied. Section~\ref{sec-BL} makes a detailed investigation
of the notion of influence introduced by Ben-Or and Linial and its relation to the auto-correlation based notion of influence. 
A discussion of the new results in this paper and their importance is given in Section~\ref{sec-discuss}.
Finally, Section~\ref{sec-conclu} concludes the paper.

\section{Background and Notation \label{sec-prelim} }
Let $\mathbb{F}_2=\{0,1\}$ denote the finite field consisting of two elements with addition represented by $\oplus$ and multiplication by $\cdot$; often, for 
$x,y\in \mathbb{F}_2$, the product $x\cdot y$ will be written as $xy$. 

	By $[n]$ we will denote the set $\{1,\ldots,n\}$.  
	For $\mathbf{x}=(x_1,\ldots,x_n)\in \mathbb{F}_2^n$, the support of $\mathbf{x}$ will be denoted by $\sym{supp}(\mathbf{x})$ which is the set $\{i:x_i=1\}$;
	the weight of $\mathbf{x}$ will be denoted by $\sym{wt}(\mathbf{x})$ and is equal to $\#\sym{supp}(\mathbf{x})$. 
	For $i\in [n]$, $\mathbf{e}_i$ denotes the vector in $\mathbb{F}_2^n$ whose $i$-th component is 1 and all other components are 0.
	By $\mathbf{0}_n$ and $\mathbf{1}_n$ we will denote the all-zero and all-one vectors of length $n$ respectively.
	For $\mathbf{x}=(x_1,\ldots,x_2),\mathbf{y}=(y_1,\ldots,y_n)\in \mathbb{F}_2^n$, we write $\mathbf{x}\leq \mathbf{y}$ if 
		$x_i=1$ implies $y_i=1$ for $i=1,\ldots,n$.
	The inner product $\langle \mathbf{x},\mathbf{y}\rangle$ of $\mathbf{x}$ 
		and $\mathbf{y}$ is defined to be $\langle \mathbf{x},\mathbf{y}\rangle =x_1y_1\oplus \cdots \oplus x_ny_n$. 
	For a subspace $E$ of $\mathbb{F}_2^n$, $E^{\perp}$ will denote the subspace $\{\mathbf{x}\in\mathbb{F}_2^n:\langle \mathbf{x},\mathbf{y}\rangle=\mathbf{0}_n,
	\mbox{ for all }\mathbf{y}\in E\}$.
	For $T\subseteq [n]$, $\chi_T$ denotes the vector in $\mathbb{F}_2^n$ where the $i$-th component of $\chi_T$ is 1 if and only if $i\in T$; further,
	$\overline{T}$ will denote the set $[n]\setminus T$.

An $n$-variable Boolean function $f$ is a map $f:\mathbb{F}_2^n\rightarrow \mathbb{F}_2$. Variables will be written in upper case and vector of variables
in bold upper case. For $\mathbf{X}=(X_1,\ldots,X_n)$, an $n$-variable Boolean function $f$ will be written as $f(\mathbf{X})$. 
The support of a Boolean function $f$ will be denoted by $\sym{supp}(f)$ which is the set $\{\mathbf{x}:f(\mathbf{x})=1\}$; the weight of $f$ will be denoted by $\sym{wt}(f)$ and 
is equal to $\#\sym{supp}(f)$. The expectation of $f$, denoted as $\mathbb{E}(f)$ (taken over a uniform random choice of $\mathbf{x}\in \mathbb{F}_2^n$), is equal 
to $\sym{wt}(f)/2^n$.
The function $f$ is said to be balanced if $\sym{wt}(f)=2^{n-1}$, i.e., $\mathbb{E}_{\mathbf{x}\in\mathbb{F}_2^n}(f)=1/2$. Noting that $f^2=f$, the
variance of $f$, denoted as $\sym{Var}(f)$ is equal to $\mathbb{E}(f^2)-\mathbb{E}(f)^2=\mathbb{E}(f)(1-\mathbb{E}(f))$.

\begin{remark}\label{rem-bool}
	In the literature, $n$-variable Boolean functions have variously been considered to be maps from $\{-1,1\}^n$ to $\{-1,1\}$, or maps from $\{-1,1\}^n$ to $\{0,1\}$,
	or maps from $\{0,1\}^n$ to $\{-1,1\}$. As stated above, in this paper, we will consider Boolean functions to be maps from $\mathbb{F}_2^n$ to $\mathbb{F}_2$.
	Results stated in this representation will be somewhat different from, though equivalent to, the results stated in the other representations.
\end{remark}

Let $\mathbf{X}=(X_1,\ldots,X_n)$ be a vector of variables and suppose $\emptyset \neq T=\{i_1,\ldots,i_t\}\subseteq [n]$, where $i_1\leq \cdots\leq i_t$. 
By $\mathbf{X}_T$ we denote the vector of variables $(X_{i_1},\ldots,X_{i_t})$.
Suppose $f(\mathbf{X})$ is an $n$-variable Boolean function. For $\bm{\alpha}\in\mathbb{F}_2^t$, 
by $f_{\mathbf{X}_T\leftarrow \bm{\alpha}}(\mathbf{X}_{\overline{T}})$ we denote the Boolean function on $n-t$ variables
obtained by setting the variables in $\mathbf{X}_T$ to the respective values in $\bm{\alpha}$. 
The function $f$ is said to be degenerate on the set of variables
$\{X_{i_1},\ldots,X_{i_t}\}$ if these variables do not influence the output of the function $f$, i.e., for any $\bm{\alpha},\bm{\beta}\in \mathbb{F}_2^t$ if we set
$f_{\bm{\alpha}}=f_{\mathbf{X}_T\leftarrow \bm{\alpha}}(\mathbf{X}_{\overline{T}})$ and 
$f_{\bm{\beta}}=f_{\mathbf{X}_T\leftarrow \bm{\beta}}(\mathbf{X}_{\overline{T}})$, then 
the functions $f_{\bm{\alpha}}$ and $f_{\bm{\beta}}$ are equal. 


Let $\psi:\mathbb{F}_2^n\rightarrow \mathbb{R}$. The Fourier transform of $\psi$ is a map $\widehat{\psi}:\mathbb{F}_2^n\rightarrow \mathbb{R}$ which is defined as follows.
\begin{eqnarray}
	\widehat{\psi}(\bm{\alpha}) 
	& = & \frac{1}{2^n} \sum_{\mathbf{x}\in\mathbb{F}_2^n} \psi(\mathbf{x})(-1)^{\langle \mathbf{x},\bm{\alpha} \rangle}. \label{eqn-fourier}
\end{eqnarray}
Given $\widehat{\psi}$, it is possible to recover $\psi$ using the following inverse formula.
\begin{eqnarray}
	\psi(\mathbf{x}) & = & \sum_{\bm{\alpha}\in\mathbb{F}_2^n} \widehat{\psi}(\bm{\alpha}) (-1)^{\langle \mathbf{x},\bm{\alpha} \rangle}. \label{eqn-inv-fourier}
\end{eqnarray}
The Poisson summation formula (see Page~77 of~\cite{carlet2021boolean}) provides a useful relation between a function $\psi:\mathbb{F}_2^n\rightarrow \mathbb{R}$
and its Fourier transform. Let $E$ be a subspace of $\mathbb{F}_2^n$ and $\mathbf{a},\mathbf{b}\in \mathbb{F}_2^n$. Then
\begin{eqnarray}\label{poisson-first-order}
\sum_{\mathbf{w}\in\mathbf{a}+E}(-1)^{\langle\mathbf{b},\mathbf{w}\rangle}\widehat{\psi}(\mathbf{w})
	& = & \frac{\#E}{2^n}(-1)^{\langle\mathbf{a},\mathbf{b}\rangle}\sum_{\mathbf{u}\in \mathbf{b}+E^{\perp}}(-1)^{\langle\mathbf{a},\mathbf{u}\rangle}\psi(\mathbf{u}).
\end{eqnarray}

The (normalised) Walsh transform of a Boolean function $f$ is a map $W_f:\mathbb{F}_2^n\rightarrow [-1,1]$ which is defined as follows. 
\begin{eqnarray}
	W_f(\bm{\alpha})
	& = & \frac{1}{2^n} \sum_{\mathbf{x}\in\mathbb{F}_2^n} (-1)^{f(\mathbf{x}) \oplus \langle \mathbf{x},\bm{\alpha} \rangle} 
	= 1-\frac{\sym{wt}(f(\mathbf{X})\oplus \langle \bm{\alpha},\mathbf{X} \rangle)}{2^{n-1}}. \label{eqn-wt-fou}
\end{eqnarray}
In other words, the Walsh transform of $f$ is the Fourier transform of $(-1)^f$. 

Note that $W_{f}(\bm{\alpha})=0$ if and only if the function $f(\mathbf{X})\oplus \langle \bm{\alpha},\mathbf{X} \rangle$ is balanced.
From Parseval's theorem (see Page~79 of~\cite{carlet2021boolean}), it follows that 
\begin{eqnarray} \label{eqn-parseval}
	\sum_{\bm{\alpha}\in \mathbb{F}_2^n} \left(W_{f}(\bm{\alpha})\right)^2 & = & 1.
\end{eqnarray}
So the values $\left\{\left(W_{f}(\bm{\alpha})\right)^2\right\}$ can be considered to be a probability distribution on $\mathbb{F}_2^n$, which assigns to 
$\bm{\alpha}\in\mathbb{F}_2^n$, the probability $\left(W_{f}(\bm{\alpha})\right)^2$. For $k\in \{0,\ldots,n\}$, let 
\begin{eqnarray} \label{eqn-pf}
	\widehat{p}_f(k) & = & \sum_{\{\mathbf{u}\in\mathbb{F}_2^n:\sym{wt}(\mathbf{u})=k\}} \left(W_{f}(\mathbf{u})\right)^2
\end{eqnarray}
be the probability assigned by the Fourier transform of $f$ to the integer $k$. Note that $\widehat{p}_f(\mathbf{0}_n)=\left(1-2\mathbb{E}(f)\right)^2$
and so $1-\widehat{p}_f(\mathbf{0}_n)=4\mathbb{E}(f)(1-\mathbb{E}(f))=4\,\sym{Var}(f)$.

The (normalised) auto-correlation function of $f$ is a map $C_f:\mathbb{F}_2^n\rightarrow [-1,1]$ defined as follows. 
\begin{eqnarray}
	\lefteqn{C_f(\bm{\alpha})} \nonumber \\
	& = & \frac{1}{2^n} \sum_{\mathbf{x}\in\mathbb{F}_2^n} (-1)^{f(\mathbf{x}) \oplus f(\mathbf{x}\oplus \bm{\alpha})} 
	= 1 - \frac{\sym{wt}( f(\mathbf{X})\oplus f(\mathbf{X} \oplus \bm{\alpha}) )}{2^{n-1}}
	= 1 - 2\Pr_{\mathbf{x}\in\mathbb{F}_2^n}[f(\mathbf{x})\neq f(\mathbf{x}\oplus \bm{\alpha})]. \label{eqn-ac-wt}
\end{eqnarray}
Note that $C_f(\mathbf{0})=1$. 

For a Boolean function $f$, the Wiener-Khintchine formula (see Page~80 of~\cite{carlet2021boolean}) relates the Walsh transform to the auto-correlation function.
\begin{eqnarray}
	\left(W_{f}(\bm{\alpha})\right)^2 & = & \widehat{C}_f(\bm{\alpha}) 
	= \frac{1}{2^n} \sum_{\mathbf{x}\in\mathbb{F}_2^n} (-1)^{\langle\bm{\alpha},\mathbf{x}\rangle} C_f(\mathbf{x}). \label{eqn-f-ac}
\end{eqnarray}
Applying the inverse Fourier transform given by~\eqref{eqn-inv-fourier} to $\widehat{C}_f(\bm{\alpha})$, we obtain
\begin{eqnarray}
	C_f(\mathbf{x}) & = & \sum_{\bm{\alpha}\in\mathbb{F}_2^n} \left(W_{f}(\bm{\alpha})\right)^2 (-1)^{\langle\bm{\alpha},\mathbf{x}\rangle}. \label{eqn-ac-wf}
\end{eqnarray}

Applying~\eqref{poisson-first-order} with $\psi=C_f$ and $\mathbf{a}=\mathbf{b}=\mathbf{0}_n$ and then using~\eqref{eqn-f-ac}, we obtain the following result
(see Proposition~5 of~\cite{DBLP:conf/eurocrypt/CanteautCCF00}).
\begin{eqnarray}
	\sum_{\mathbf{w}\in E} \left(W_{f}(\mathbf{w})\right)^2 & = & \frac{\#E}{2^n} \sum_{\mathbf{u}\in E^{\perp}} C_f(\mathbf{u}).  \label{eqn-f-ac-subspace}
\end{eqnarray}

Let $T\subseteq [n]$ with $\#T=t$ and for $\bm{\alpha}\in \mathbb{F}_2^{n-t}$, let $f_{\bm{\alpha}}$ denote 
$f_{\mathbf{X}_{\overline{T}}\leftarrow \bm{\alpha}}(\mathbf{X}_{T})$.
Then $f_{\bm{\alpha}}$ is a $t$-variable function. From the second order Poisson summation formula (see Page~81 of~\cite{carlet2021boolean} for the general statement
of this result), we have
\begin{eqnarray}\label{eqn-f-ac-subspace-restrict}
	\sum_{\mathbf{w}\leq \chi_{\overline{T}}} \left(W_{f}(\mathbf{w})\right)^2 
	& = & \frac{1}{2^{n-t}} \sum_{\bm{\alpha}\leq \chi_{\overline{T}}} \left(W_{f_{\alpha}}(\mathbf{0}_t) \right)^2.
\end{eqnarray}

\begin{remark}
	We have normalised the Walsh transform and the auto-correlation function by $2^n$ so that the values lie in the range $[-1,1]$. The non-normalised versions
	have also been used in the literature. We note in particular that~\cite{carlet2021boolean} uses the non-normalised versions.
	When we use results from~\cite{carlet2021boolean}, we normalise them appropriately.
\end{remark}

\paragraph{Some Boolean function classes.} 
Let $f$ be an $n$-variable Boolean function.
\begin{itemize} 
	\item The function $f$ is said to be bent~\cite{DBLP:journals/jct/Rothaus76} if $W_{f}(\bm{\alpha})=\pm 2^{-n/2}$ for all $\bm{\alpha}\in \mathbb{F}_2^n$. Bent functions
		exist if and only if $n$ is even. 
	\item The function $f$ is said to satisfy propagation characteristics~\cite{preneel1990propagation} of degree $k\geq 1$, written as PC($k$) if $C_f(\mathbf{u})=0$ for 
all $\mathbf{u}\in\mathbb{F}_2^n$ with $1\leq \sym{wt}(\mathbf{u})\leq k$. 
	\item The function $f$ is said to be $m$-resilient~\cite{DBLP:journals/tit/Siegenthaler84,DBLP:journals/tit/XiaoM88} if $W_{f}(\bm{\alpha})=0$ for 
		all $\bm{\alpha}\in \mathbb{F}_2^n$ with $0\leq \sym{wt}(\bm{\alpha})\leq m$.
	\item The function $f$ is said to be an $s$-junta if there is a subset $S\subseteq [n]$ with $\#S\leq s$ such that $f$ is degenerate on the variables
		indexed by $\overline{S}$.
\end{itemize}

\subsection{Influence \label{subsec-influence}}
Let $f(\mathbf{X})$ be an $n$-variable Boolean function where $\mathbf{X}=(X_1,\ldots,X_n)$. For $i\in [n]$, the influence
of $X_i$ on $f$ is denoted by $\sym{inf}_f(i)$ and is defined to be the probability (over a uniform random choice of $\mathbf{x}\in\mathbb{F}_2^n$) that $f(\mathbf{x})$ is 
not equal to $f(\mathbf{x}\oplus \mathbf{e}_i)$, i.e.,
\begin{eqnarray}\label{eqn-inf-sing}
	\sym{inf}_f(i) & = & \Pr_{\mathbf{x}\in\mathbb{F}_2^n}[f(\mathbf{x}) \neq f(\mathbf{x}\oplus \mathbf{e}_i)]. 
\end{eqnarray}
The total influence $\sym{inf}(f)$ of the individual variables is defined to be the sum of the influences of the individual variables, i.e.
$\sym{inf}(f)=\sum_{i\in [n]}\sym{inf}_f(i)$. 

Let $f$ be an $n$-variable Boolean function and $\emptyset \neq T\subseteq [n]$ with $t=\#T$. The influence of the set of variables indexed by $T$ on $f$ has been defined
in the literature in four different ways. These definitions are given below.

\paragraph{Ben-Or and Linial~\cite{ben1987collective}.}
The definition of influence introduced in~\cite{ben1987collective} is the following. 
\begin{eqnarray}\label{eqn-BL-inf}
	\mathcal{I}_f(T) & = & \Pr_{\bm{\alpha}\in \mathbb{F}_2^{n-t}}\left[f_{\mathbf{X}_{\overline{T}}\leftarrow \bm{\alpha}}(\mathbf{X}_{{T}}) \mbox{ is not constant}\right].
\end{eqnarray}

\paragraph{Fischer et al.~\cite{fischer02} and Blais~\cite{blais2009testing}.}
The same quantity has been defined in two different ways in Fischer et al.~\cite{fischer02} and Blais~\cite{blais2009testing}. In~\cite{fischer02}, this
quantity was called `variation' and in~\cite{blais2009testing}, it was termed `influence'. Here we provide the formulation as given in~\cite{blais2009testing}.
For $\mathbf{x},\mathbf{y}\in \mathbb{F}_2^n$, let $Z(T,\mathbf{x},\mathbf{y})$ denote the vector $\mathbf{z}\in\mathbb{F}_2^n$, where
$z_i=y_i$, if $i\in T$ and $z_i=x_i$ otherwise. The definition of influence given in~\cite{blais2009testing} is the following.
\begin{eqnarray}\label{eqn-blais}
	I_{f}(T) & = & \Pr_{\mathbf{x},\mathbf{y}\in \mathbb{F}_2^n} \left[ f(\mathbf{x}) \neq f(Z(T,\mathbf{x},\mathbf{y})) \right].
\end{eqnarray}

\paragraph{Gangopadhyay and St{\u{a}}nic{\u{a}}~\cite{DBLP:journals/iacr/GangopadhyayS14}.} 
The definition of influence introduced in~\cite{DBLP:journals/iacr/GangopadhyayS14} is the following. 
\begin{eqnarray}\label{eqn-GS-inf}
	\mathcal{J}_f(T) & = & \Pr_{\mathbf{x}\in\mathbb{F}_2^n}[f(\mathbf{x})\neq f(\mathbf{x}\oplus \chi_T)] = \frac{1}{2}\left(1-C_f(\chi_T) \right).
\end{eqnarray}

\paragraph{Tal~\cite{tal2017tight}.}
For $\bm{\beta}\in \mathbb{F}_2^t$, let $f_{\bm{\beta}}$ denote the function $f_{\mathbf{X}_{T}\leftarrow \bm{\beta}}$.
Let $D_Tf:\{0,1\}^{n-t}\rightarrow [-1,1]$ be defined as follows. For $\mathbf{y}\in F_2^{n-t}$,
$(D_Tf)(\mathbf{y}) = 1/2^t \times \sum_{\bm{\beta} \in \mathbb{F}_2^t} (-1)^{\sym{wt}(\bm{\beta}) + f_{\bm{\beta}}(\mathbf{y})}$.
The definition of influence given in~\cite{tal2017tight} is the following.
\begin{eqnarray}\label{eqn-inf-tal}
	J_{f}(T) & = & \displaystyle \mathop{\mathbb{E}}_{\mathbf{y}\in \mathbb{F}_2^{n-t}} \left[\left(D_Tf(\mathbf{y})\right)^2\right].
\end{eqnarray}

\section{Influence from Auto-Correlation \label{sec-inf-set} }
The auto-correlation function is a very useful tool for expressing various properties of Boolean functions. We refer to~\cite{carlet2021boolean} for the many uses of the
auto-correlation function in the context of cryptographic properties of Boolean functions. 
Given an $n$-variable Boolean function $f$ and $\bm{\alpha}\in \mathbb{F}_2^n$, the value of the auto-correlation function $C_f$ at 
$\bm{\alpha}$, i.e., $C_f(\bm{\alpha})$ is the number of places $f(\mathbf{X})$ and $f(\mathbf{X}\oplus \bm{\alpha})$ are equal minus the number of places they are
unequal (normalised by $2^n$). So the auto-correlation function at $\bm{\alpha}$ to some extent captures the effect on $f$ of flipping all the bits in the support of
$\bm{\alpha}$. This suggests that the auto-correlation function is an appropriate mechanism to capture the influence of a set of variables on a Boolean function.
We note that for $i\in [n]$, $\sym{inf}_f(i)$ can be written as follows.
\begin{eqnarray}\label{eqn-inf-sing-ac}
	\sym{inf}_f(i) & = & \frac{1}{2}\left(1-C_f(\mathbf{e}_i) \right) = 1-\frac{1}{2}\left(C_f(\mathbf{0}) + C_f(\mathbf{e}_i) \right).
\end{eqnarray}

Let $f(X_1,\ldots,X_n)$ be an $n$-variable Boolean function and $\emptyset \neq T=\{i_1,\ldots,i_t\}\subseteq [n]$. 
We denote the influence of the set of variables $\{X_{i_1},\ldots,X_{i_t}\}$ corresponding to 
$T=\{i_1,\ldots,i_t\}$ on the Boolean function $f$ by $\sym{inf}_f(T)$. 
Following the auto-correlation based expression of the influence of a single variable on 
a Boolean function given by~\eqref{eqn-inf-sing-ac}, we put forward the following definition of $\sym{inf}_f(T)$.
\begin{eqnarray}\label{eqn-inf-set}
	\sym{inf}_f(T) & = & 1 - \frac{1}{2^{\#T}}\left(\sum_{\bm{\alpha}\leq \chi_T} C_f(\bm{\alpha})\right).
\end{eqnarray}
It is easy to note that for a singleton set $T=\{i\}$, $\sym{inf}_f(T)=\sym{inf}_f(i)$. Further, one may note that
$\sym{inf}_f(T)=2^{1-t} \times \sum_{S\subseteq T} \mathcal{J}_f(S)$. 

\begin{remark}\label{rem-Teq1}
	We note that $\sym{inf}_f(T)$, $\mathcal{I}_f(T)$, $J_f(T)$ and $\mathcal{J}_f(T)$ (defined in Section~\ref{subsec-influence}) agree with each other when $\#T=1$. Also, we 
	later show that $I_f(T)=\sym{inf}_f(T)/2$.
\end{remark}

It is perhaps not immediately obvious that the definition of influence given by~\eqref{eqn-inf-set} is appropriate. We later show in Theorem~\ref{thm-inf-basic}
that this definition satisfies a set of intuitive desiderata that any notion of influence may be expected to satisfy.

Let $f$ be an $n$-variable function and $t$ be an integer with $1\leq t\leq n$. Then the $t$-influence of $f$ is the total influence (scaled by ${n\choose t}$)
obtained by summing the influence of every set of $t$ variables on the function $f$, i.e., 
\begin{eqnarray}\label{eqn-t-inf}
	t\mbox{-}\sym{inf}(f) & = & \frac{\sum_{\{T\subseteq [n]: \#T=t\}} \sym{inf}_f(T)}{{n\choose t}}.
\end{eqnarray}
Note that $1\mbox{-}\sym{inf}(f)$ is equal to $\sym{inf}(f)/n$, i.e., $1\mbox{-}\sym{inf}(f)$ is the sum of the influences of the individual
variables scaled by a factor of $n$.

The following result provides a characterisation of influence in terms of the Walsh transform.
\begin{theorem}\label{thm-inf-fourier}
	Let $f$ be an $n$-variable Boolean function and $\emptyset \neq T\subseteq [n]$. Then
	\begin{eqnarray}\label{eqn-inf-fourier}
		\sym{inf}_f(T) & = & \sum_{\{\mathbf{u}\in\mathbb{F}_2^n:\sym{supp}(\mathbf{u})\cap T\neq \emptyset\}} \left(W_{f}(\mathbf{u})\right)^2.
	\end{eqnarray}
\end{theorem}
\begin{proof}
	Let $\#T=t$. Let $E$ be the subspace $\{\mathbf{x}\in\mathbb{F}_2^n: \mathbf{x}\leq \chi_{\overline{T}}\}$. Then
	$\#E=2^{n-t}$ and $E^{\perp}=\{\mathbf{y}\in\mathbb{F}_2^n: \mathbf{y}\leq \chi_T\}$. 
	Using~\eqref{eqn-f-ac-subspace}, we obtain
	\begin{eqnarray}
		\sum_{\mathbf{x}\leq \chi_{\overline{T}}} \left(W_{f}(\mathbf{x})\right)^2 
		& = & \frac{2^{n-t}}{2^n} \sum_{\mathbf{y}\leq \chi_T} C_f(\mathbf{y}) = \frac{1}{2^{\#T}} \sum_{\mathbf{y}\leq \chi_T} C_f(\mathbf{y}). \label{eqn-thm1eqn1}
	\end{eqnarray}
	Using~\eqref{eqn-thm1eqn1} with~\eqref{eqn-inf-set} and~\eqref{eqn-parseval} we have
	\begin{eqnarray*}
\sym{inf}_f(T) & = & 1 - \sum_{\mathbf{x}\leq \chi_{\overline{T}}} \left(W_{f}(\mathbf{x})\right)^2 
		= \sum_{\mathbf{w}\in\mathbb{F}_2^n} \left(W_{f}(\mathbf{w})\right)^2 - \sum_{\mathbf{x}\leq \chi_{\overline{T}}} \left(W_{f}(\mathbf{x})\right)^2
		= \sum_{\mathbf{u}\not\leq \chi_{\overline{T}}} \left(W_{f}(\mathbf{u})\right)^2.
	\end{eqnarray*}
	The condition $\mathbf{u}\not\leq \chi_{\overline{T}}$ is equivalent to $\sym{supp}(\mathbf{u})\cap T\neq \emptyset$. 
\end{proof}

It is a well known result (see Page~52 of~\cite{o2014analysis}) that for an $n$-variable Boolean function, the total influence of the individual variables, i.e., 
$\sym{inf}(f)$ is the expected value of a random variable 
which takes the value $k$ with probability $\widehat{p}_f(k)$ for $k=0,\ldots,n$. We generalise this result to the case of $t\mbox{-}\sym{inf}(f)$ for $t\geq 1$. 

For positive integers $n$, $t$ and $k$ with, $1\leq t\leq n$ and $0\leq k\leq n$, fix a subset $S$ of $[n]$ with $\#S=k$ and let $N_{n,t,k}$ be the number of subsets of 
$[n]$ of size $t$ which contains at least one element of $S$. Then
\begin{eqnarray}\label{eqn-N}
	N_{n,t,k} & = & {n\choose t} - {n-k\choose t} = \sum_{i=1}^{\min(k,t)}{k\choose i}{n-k\choose t-i}.
\end{eqnarray}
It follows that $N_{n,t,0}=0$, $N_{n,t,k}={n\choose t}$ for $n-t+1\leq k\leq n$, and $N_{n,1,k}=k$ for $k=0,\ldots,n$.

\begin{theorem}\label{thm-inf-total}
	Let $f$ be an $n$-variable function and $t\in [n]$. Then 
	\begin{eqnarray}\label{eqn-t-inf-exp}
		t\mbox{-}\sym{inf}(f) & = & \frac{1}{{n\choose t}} \sum_{k=1}^n N_{n,t,k}\  \widehat{p}_f(k)  = \frac{1}{{n\choose t}}\mathbb{E}[Z],
	\end{eqnarray}
	where $Z$ is the number of $t$-element subsets of $[n]$ which have a non-empty intersection with a set $S\subseteq [n]$ chosen with probability $(W_f(\chi_S))^2$.
\end{theorem}
\begin{proof}
	We start with the proof of the first equality in~\eqref{eqn-t-inf-exp}.
	Consider $\mathbf{u}\in\mathbb{F}_2^n$ with $\#\sym{supp}(\mathbf{u})=k$. For $1\leq i\leq \min(k,t)$, the number of
        subsets $T$ of $[n]$ of cardinality $t$ whose intersection with $\sym{supp}(\mathbf{u})$ is of size $i$ is ${k\choose i}{n-k\choose t-i}$. Summing over $i$
        provides the number of subsets $T$ of $[n]$ of cardinality $t$ with which $\sym{supp}(\mathbf{u})$ has a non-empty intersection. 
	From~\eqref{eqn-t-inf} and Theorem~\ref{thm-inf-fourier}, we have
        \begin{eqnarray}
                t\mbox{-}\sym{inf}(f) & = & \frac{1}{{n\choose t}} 
                        \sum_{k=1}^n \sum_{\{\mathbf{u}\in\mathbb{F}_2^n:\sym{wt}(\mathbf{u})=k\}} \sum_{i=1}^{\min(k,t)} {k\choose i}{n-k\choose t-i} 
                \left(W_{f}(\mathbf{u}) \right)^2 \nonumber \\
                & = & \frac{1}{{n\choose t}} 
                        \sum_{k=1}^n \sum_{i=1}^{\min(k,t)} {k\choose i}{n-k\choose t-i} \sum_{\{\mathbf{u}\in\mathbb{F}_2^n:\sym{wt}(\mathbf{u})=k\}} 
                \left(W_{f}(\mathbf{u})\right)^2 \nonumber \\
		& = & \frac{1}{{n\choose t}} \sum_{k=1}^n N_{n,t,k} \widehat{p}_f(k) \nonumber \\ 
		& = & \frac{1}{{n\choose t}}\mathbb{E}[Z]. \nonumber 
        \end{eqnarray}

	The second equality in~\eqref{eqn-t-inf-exp} follows from the observation that if $\#S=k$, then $Z=N_{n,t,k}$. 
\end{proof} 
Poincar\'{e} inequality (see Page~52 of~\cite{o2014analysis}) states that the total influence of the individual variables, i.e., $\sym{inf}(f)$
is bounded below by $4\,\sym{Var}(f)$. We obtain a generalisation of this result as a corollary of Theorem~\ref{thm-inf-total}.
\begin{corollary}\label{cor-iso}
	Let $f$ be an $n$-variable Boolean function and $t\in [n]$. Then
	\begin{eqnarray}\label{eqn-iso}
		t\mbox{-}\sym{inf}(f) & \geq & \frac{4t}{n}\,\sym{Var}(f).
	\end{eqnarray}
	Equality is achieved for $t=n$.
\end{corollary}
\begin{proof}
	Note that for $0\leq k\leq n$, $n-i+1>0$ for $2\leq i\leq k$ and so $1-t/(n-i+1)<1$ for $2\leq i\leq k$. Using this, we have 
	\begin{eqnarray*}
		\frac{{n-k\choose t}}{{n\choose t}} = \left(1-\frac{t}{n}\right)\left(1-\frac{t}{n-1}\right)\cdots \left(1-\frac{t}{n-k+1} \right) \leq 1 - \frac{t}{n}.
	\end{eqnarray*}
	It follows that for $k\in [n]$, $N_{n,t,k}/{n\choose t} \geq t/n$, where equality is achieved for $t=n$. 
	So from~\eqref{eqn-t-inf-exp},
	\begin{eqnarray*}
		t\mbox{-}\sym{inf}(f) & \geq & \frac{t}{n} \sum_{k=1}^n \widehat{p}_f(k) = \frac{t}{n} (1-\widehat{p}_f(\mathbf{0}_n)) = \frac{4t}{n} \sym{Var}(f).
	\end{eqnarray*}
\end{proof}

The Fourier/Walsh transform based expression for the total influence given by Theorem~\ref{thm-inf-total} is a useful result.
Corollary~\ref{cor-iso} above provides a direct application of Theorem~\ref{thm-inf-total}. In Theorem~\ref{thm-max-t-inf}, proved later, we use the expression given by 
Theorem~\ref{thm-inf-total} to characterise the functions which achieve the maximum value of the total influence as resilient functions. 
In Theorem~\ref{thm-mi}, also proved later, the expression is used to
show that total influence is monotonic increasing in $t$. An additional application of Theorem~\ref{thm-inf-total} is given next. 

Given an $n$-variable Boolean function $f$, we say that the Fourier spectrum of $f$ is $\epsilon$-concentrated on coefficients of weights up to $k$
if $\sum_{i\geq k}\widehat{p}_i(f) \leq \epsilon$.
Proposition~3.2 on Page~69 of~\cite{o2014analysis} shows that the Fourier spectrum of $f$ is $\epsilon$-concentrated on coefficients of weights up to $k$,
where $k$ is the least positive integer such that $k\geq n\times 1\mbox{-}\sym{inf}(f)/\epsilon$ and 
$1\mbox{-}\sym{inf}(f)\leq \epsilon\leq 1$. The following theorem generalises this result to arbitrary values of $t$.
\begin{theorem} \label{thm-conc}
	For any $n$-variable Boolean function, $t\in [n]$ and $\epsilon\in [t\mbox{-}\sym{inf}(f),1]$, the Fourier spectrum of $f$ is $\epsilon$-concentrated on 
	coefficients of weights up to $k_t$, where $k_t$ is the least positive integer such that 
	\begin{eqnarray}\label{eqn-k-conc}
		k_t & \geq & t-1+(n-t+1)\left(1-(1-x_t)^{1/t}\right),
	\end{eqnarray}
	and $x_t=\frac{t\mbox{-}\sym{inf}(f)}{\epsilon}$.
\end{theorem}
\begin{proof}
	The condition given by~\eqref{eqn-k-conc} holds if and only if $(n-k_t)\leq (n-t+1)(1-x_t)^{t}$ which holds if and only if 
	\begin{eqnarray}\label{eqn-tmp}
		\frac{(n-k_t)^t}{t!} & \leq & \frac{(n-t+1)^t}{t!} (1-x_t).
	\end{eqnarray}
	Using the inequalities ${n-k_t \choose t}\leq (n-k_t)^t/t!$ and $(n-t+1)^t/t!\leq {n\choose t}$, from~\eqref{eqn-tmp} we obtain
		${n-k_t\choose t} \leq {n\choose t}(1-x_t)$, which holds if and only if 
        \begin{eqnarray}\label{eqn-tmp0}
		{n\choose t} - {n-k_t\choose t} & \geq & {n\choose t}x_t = {n\choose t} \frac{t\mbox{-}\sym{inf}(f)}{\epsilon}.
        \end{eqnarray}
	Let if possible that the Fourier transform of $f$ is not $\epsilon$-concentrated on coefficients of weights up to $k_t$. Then for $k_t$ satisfying~\eqref{eqn-tmp0},
	we have $\sum_{k=k_t}^n\widehat{p}_k(f) > \epsilon$. From~\eqref{eqn-t-inf-exp}, we have 
	\begin{eqnarray}
		t\mbox{-}\sym{inf}(f) & = & \frac{1}{{n\choose t}} \sum_{k=1}^n N_{n,t,k}\  \widehat{p}_f(k)  \nonumber \\
		& = & \frac{1}{{n\choose t}} \left(\sum_{k=1}^{k_t-1} N_{n,t,k}\  \widehat{p}_f(k) + \sum_{k=k_t}^n N_{n,t,k}\  \widehat{p}_f(k) \right) \nonumber \\
		& \geq & \frac{1}{{n\choose t}} \sum_{k=k_t}^n \left({n\choose t} - {n-k\choose t} \right) \widehat{p}_f(k) \quad \mbox{(using~\eqref{eqn-N})} \nonumber \\
		& \geq & \frac{1}{{n\choose t}} \left({n\choose t} - {n-k_t\choose t} \right) \sum_{k=k_t}^n \widehat{p}_f(k) \nonumber \\
		& > & \frac{1}{{n\choose t}} \left({n\choose t} - {n-k_t\choose t} \right) \epsilon \quad \mbox{(by assumption)} \nonumber \\
		& \geq & t\mbox{-}\sym{inf}(f) \quad \mbox{(using~\eqref{eqn-tmp0}).} \nonumber
        \end{eqnarray}
This gives us the desired contradiction.
\end{proof}

An alternative Walsh transform based characterisation of influence is given by the following result.

\begin{theorem}\label{thm-inf-ac-exp}
	Let $f$ be an $n$-variable function and $\emptyset \neq T\subseteq [n]$. Then
	\begin{eqnarray}\label{eqn-inf-ac-exp}
		\sym{inf}_f(T) & = & 1- \frac{1}{2^{n-t}} \sum_{\bm{\alpha}\in \mathbb{F}_2^{n-t}} \left( W_{{f}_{\bm{\alpha}}}(\mathbf{0}_t) \right)^2, 
	\end{eqnarray}
	where $f_{\bm{\alpha}}$ denotes $f_{\mathbf{X}_{\overline{T}}\leftarrow \bm{\alpha}}$. 
\end{theorem}
\begin{proof}
	Let $\#T=t$. Let $E=\{\mathbf{x}\in\mathbb{F}_2^n: \mathbf{x}\leq \chi_{\overline{T}}\}$ and so
	$E^{\perp}=\{\mathbf{x}\in\mathbb{F}_2^n: \mathbf{x}\leq \chi_{T}\}$.
	Using~\eqref{eqn-f-ac-subspace} and~\eqref{eqn-f-ac-subspace-restrict} we have
	\begin{eqnarray*}
		\frac{1}{2^t}\sum_{\mathbf{u}\leq \chi_T} C_f(\mathbf{u}) 
		& = & \frac{1}{2^{n-t}} \sum_{\bm{\alpha}\in \mathbb{F}_2^{n-t}} \left( W_{{f}_{\bm{\alpha}}}(\mathbf{0}_t) \right)^2.
	\end{eqnarray*}
	Using the definition of influence given in~\eqref{eqn-inf-set}, we obtain the required result.
\end{proof}
\begin{remark}\label{rem-comp}
Theorems~\ref{thm-inf-fourier} and~\ref{thm-inf-ac-exp} provide two different Walsh transform based characterisations of $\sym{inf}_f(T)$. The expression for $\sym{inf}_f(T)$ given 
	by~\eqref{eqn-inf-ac-exp}
can be computed in $O(2^n)$ time, while the expression given by~\eqref{eqn-inf-fourier} in general will require $O(n2^n)$ time using the fast Fourier transform algorithm to compute
the required values of the Walsh transform.
\end{remark}

We obtain the following corollary of Theorem~\ref{thm-inf-ac-exp}.
\begin{corollary}\label{cor-inf-ac-exp}
	Let $f$ be an $n$-variable Boolean function and $\emptyset \neq T \subseteq [n]$. Then
	\begin{eqnarray}\label{eqn-inf-ac-exp-var}
		\sym{inf}_f(T) & = & 
		\frac{1}{2^{n-2-t}}\sum_{\bm{\alpha}\in \mathbb{F}_2^{n-t}} \sym{Var}(f_{\bm{\alpha}})
	\end{eqnarray}
	where $f_{\bm{\alpha}}$ denotes $f_{\mathbf{X}_{\overline{T}}\leftarrow \bm{\alpha}}$. 
\end{corollary}
\begin{proof}
	Using~\eqref{eqn-inf-ac-exp}, we have
	\begin{eqnarray}
		\sym{inf}_f(T) 
		& = & 1- \frac{1}{2^{n-t}} \sum_{\bm{\alpha}\in \mathbb{F}_2^{n-t}} \left( W_{{f}_{\bm{\alpha}}}(\mathbf{0}_t) \right)^2 \nonumber \\
		& = & \frac{1}{2^{n-t}}\sum_{\bm{\alpha}\in \mathbb{F}_2^{n-t}} \left(1-\left( W_{{f}_{\bm{\alpha}}}(\mathbf{0}_t) \right)^2 \right) \nonumber \\
		& = & \frac{1}{2^{n-t}}\sum_{\bm{\alpha}\in \mathbb{F}_2^{n-t}} 
			4\mathbb{E}\left(f_{\bm{\alpha}} \right) \left(1-\mathbb{E}\left(f_{\bm{\alpha}} \right) \right) \nonumber \\
		& = & \frac{1}{2^{n-2-t}}\sum_{\bm{\alpha}\in \mathbb{F}_2^{n-t}} \sym{Var}(f_{\bm{\alpha}}).
	\end{eqnarray}

\end{proof}

One may consider some basic desiderata that any reasonable measure of influence should satisfy. Since we are considering normalised measures, the value of influence
should be in the set $[0,1]$ and it should take the value $0$ if and only if the function is degenerate on the set of variables. Further, by expanding a set of
variables, the value of influence should not decrease, i.e. influence should be monotonic non-decreasing. Also, sub-additivity is a desirable property. 
The following result shows these properties for $\sym{inf}_f(T)$ and also characterises the condition under which $\sym{inf}_f(T)$ takes its maximum value 1.
\begin{theorem} \label{thm-inf-basic}
        Let $f$ be an $n$-variable Boolean function and $\emptyset \neq T,S \subseteq [n]$. Then
	\begin{enumerate}
		\item $0\leq \sym{inf}_f(T) \leq 1$. 
		\item $\sym{inf}_f(T)=0$ if and only if the function $f$ is degenerate on the variables indexed by $T$.
		\item $\sym{inf}_f(T) = 1$ if and only if $f_{\bm{\alpha}}$ is balanced for each $\bm{\alpha}\in \mathbb{F}_2^{n-t}$, where $f_{\bm{\alpha}}$
			denotes $f_{\mathbf{X}_{\overline{T}}\leftarrow \bm{\alpha}}$.
		\item $\sym{inf}_f(S\cup T) \geq \sym{inf}_f(T)$. 
		\item $\sym{inf}_f(S\cup T) = \sym{inf}_f(S) + \sym{inf}_f(T) - \sum_{\mathbf{u}\in\mathcal{U}}\left(W_{f}(\mathbf{u})\right)^2$, where
			$\mathcal{U} = \{\mathbf{u}\in\mathbb{F}_2^n:\sym{supp}(\mathbf{u})\cap S\neq \emptyset \neq \sym{supp}(\mathbf{u}) \cap T\}$.
			Consequently, $\sym{inf}_f(S\cup T) \leq \sym{inf}_f(S) + \sym{inf}_f(T)$ (i.e., $\sym{inf}_f(T)$ satisfies sub-additivity).
	\end{enumerate}
\end{theorem}
\begin{proof}
	The first point follows from Theorem~\ref{thm-inf-fourier} and Parseval's theorem. The fourth and fifth points also follow from Theorem~\ref{thm-inf-fourier}.
	The third point follows from Theorem~\ref{thm-inf-ac-exp}.

	Consider the second point. From~\eqref{eqn-inf-ac-exp}, $\sym{inf}_f(T)=0$ if and only if 
	$\sum_{\bm{\alpha}\in \mathbb{F}_2^{n-t}} \left( W_{{f}_{\bm{\alpha}}}(\mathbf{0}_t) \right)^2=2^{n-t}$. Since $\left( W_{{f}_{\bm{\alpha}}}(\mathbf{0}_t) \right)^2\leq 1$,
	it follows that $\sum_{\bm{\alpha}\in \mathbb{F}_2^{n-t}} \left( W_{{f}_{\bm{\alpha}}}(\mathbf{0}_t) \right)^2=2^{n-t}$ if and only if
	$\left( W_{{f}_{\bm{\alpha}}}(\mathbf{0}_t) \right)^2=1$ (equivalently, $f_{\bm{\alpha}}$ is constant) for all $\bm{\alpha}\in \mathbb{F}_2^{n-t}$. 
	The last condition is equivalent to the statement that $f$ is degenerate on the set of variables indexed by $T$.
\end{proof}


\begin{remark} \label{rem-GS}
	For the Gangopadhyay and St{\u{a}}nic{\u{a}} notion of influence $\mathcal{J}_f(T)$ (see~\ref{eqn-GS-inf}) the second point of
	Theorem~\ref{thm-inf-basic} does not hold. It is possible that $f$ is not degenerate on the variables indexed by $T$, yet $\mathcal{J}_f(T)=0$. 
	For example, let $f(X_1,X_2,X_3,X_4)=(1\oplus X_1)X_2(X_3\oplus X_4)$ and $T=\{3,4\}$. Then it may be checked that $\mathcal{J}_f(T)=0$, but
	$f$ is not degenerate on the set of variables $\{X_3,X_4\}$ as $f(0,1,0,0)=0\neq f(0,1,0,1)$. 

	If a function is not degenerate on the set of variables indexed by $T$, then these variables have an effect on value of $f$. Any reasonable
	measure of influence should ensure that if $f$ is not degenerate on a set of variables, then the value of the measure for this set of variables is positive.
	Since this condition does not hold for $\mathcal{J}_f(T)$, this measure cannot be considered to be a satisfactory measure of influence of a set of
	variables.
\end{remark}

\begin{theorem}\label{thm-max-t-inf}
	Let $f$ be an $n$-variable Boolean function and $t$ be an integer with $1\leq t\leq n$. 
	\begin{enumerate}
		\item $t$-$\sym{inf}(f)$ takes its maximum value 1 if and only if $f$ is $(n-t)$-resilient.
		\item $t$-$\sym{inf}(f)$ takes its minimum value 0 if and only if $f$ is a constant function.
	\end{enumerate}
\end{theorem}
\begin{proof}
	From~\eqref{eqn-t-inf-exp} and recalling that $N_{n,t,0}=0$ and $N_{n,t,k}={n\choose t}$ for $n-t+1\leq k\leq n$, we have
	\begin{eqnarray}
		t\mbox{-}\sym{inf}(f) 
		& = & \frac{1}{{n\choose t}} \sum_{k=1}^n N_{n,t,k} \widehat{p}_f(k) \nonumber \\
		& = & \frac{1}{{n\choose t}} \left( \sum_{k=0}^{n-t} \left({n\choose t}-{n-k\choose t} \right) + \sum_{k=n-t+1}^n {n\choose t} \right) 
				\widehat{p}_f(k) \nonumber \\
		& = & \frac{1}{{n\choose t}} \left( \sum_{k=0}^n {n\choose t} - \sum_{k=0}^{n-t} {n-k\choose t} \right) 
				\widehat{p}_f(k) \nonumber \\
		& = & 1 - \frac{1}{{n\choose t}} \sum_{k=0}^{n-t} {n-k\choose t} \widehat{p}_f(k). \label{eqn-inf-fourier-t}
	\end{eqnarray}

	From~\eqref{eqn-inf-fourier-t}, $t$-$\sym{inf}(f)$ takes its maximum value of 1 if and only if 
		$\sum_{k=0}^{n-t} {n-k\choose t}\widehat{p}_f(k)=0$ which holds if and only if
		$\widehat{p}_f(k)=0$ for $k=0,\ldots, n-t$, i.e., if and only if $f$ is $(n-t)$-resilient. This shows the first point.

	For the second point, from~\eqref{eqn-inf-fourier-t}, $t\mbox{-}\sym{inf}(f)=0$ if and only if 
	\begin{eqnarray}
		{n\choose t} \widehat{p}_f(0) + {n-1\choose t} \widehat{p}_f(1) + \cdots + {t \choose t} \widehat{p}_f(t) = {n\choose t}. \label{eqn-tt00}
	\end{eqnarray}
	If $f$ is a constant function, then $\widehat{p}_f(0)=1$ and $\widehat{p}_f(k)=0$ for $k\in [n]$. So~\eqref{eqn-tt00} holds.
	On the other hand, if $f$ is not a constant function, then $\widehat{p}_f(0)<1$. In this case, 
	\begin{eqnarray*}
		\lefteqn{ {n\choose t} \widehat{p}_f(0) + {n-1\choose t} \widehat{p}_f(1) + \cdots + {t \choose t} \widehat{p}_f(t) } \\
		& \leq & {n\choose t} \widehat{p}_f(0) + {n-1\choose t} (\widehat{p}_f(1) + \cdots + \widehat{p}_f(n)) \\
		& = & {n\choose t} \widehat{p}_f(0) + {n-1\choose t} (1-\widehat{p}_f(0)) < {n\choose t} .
	\end{eqnarray*}
\end{proof}


The next result shows that as $t$ increases, the value of $t\mbox{-}\sym{inf}(f)$ is non-decreasing.

\begin{theorem}\label{thm-mi}
	Let $f$ be an $n$-variable Boolean function. For $t\in [n]$, $t\mbox{-}\sym{inf}(f)$ increases monotonically with $t$.
\end{theorem}
\begin{proof}
	For $t\in [n-1]$, the following calculations show that $t\mbox{-}\sym{inf}(f)$ is at most $(t+1)\mbox{-}\sym{inf}(f)$. 
	\begin{eqnarray}
		\lefteqn{t\mbox{-}\sym{inf}(f) \leq (t+1)\mbox{-}\sym{inf}(f)} \nonumber \\
		& \Longleftrightarrow & 
		1-\sum_{k=0}^{n-t} \frac{{n-k\choose t}}{{n\choose t}} \widehat{p}_f(k) \leq 1-\sum_{k=0}^{n-t-1} \frac{{n-k\choose t+1}}{{n\choose t+1}} \widehat{p}_f(k)  
			\nonumber \\
		& \Longleftrightarrow & 
		\sum_{k=0}^{n-t} \frac{{n-k\choose t}}{{n\choose t}} \widehat{p}_f(k) \geq \sum_{k=0}^{n-t-1} \frac{{n-k\choose t+1}}{{n\choose t+1}} \widehat{p}_f(k)  
			\nonumber \\
		& \Longleftrightarrow & 
		\frac{1}{{n\choose t}}\widehat{p}_f(n-t) + \sum_{k=0}^{n-t-1} \left(\frac{{n-k\choose t}}{{n\choose t}}-\frac{{n-k\choose t+1}}{{n\choose t+1}}\right) 
			\widehat{p}_f(k) \geq 0 \nonumber \\
		& \Longleftrightarrow & 
		\frac{1}{{n\choose t}}\widehat{p}_f(n-t) + \sum_{k=0}^{n-t-1} \left(\frac{(n-k)!(n-t-1)!}{n!(n-k-t-1)!} \frac{k}{n-t-k} \right) 
			\widehat{p}_f(k) \geq 0. \label{eqn-t00a}
	\end{eqnarray}
	For $k$ in the range $0$ to $n-t-1$, it follows that $k/(n-t-k) \geq 0$. So the relation in~\eqref{eqn-t00a} holds showing that
	$t\mbox{-}\sym{inf}(f) \leq (t+1)\mbox{-}\sym{inf}(f)$.

\end{proof}

\subsection{Geometric Interpretation \label{subsec-path} }
Let $H_n$ be the $n$-dimensional hypercube, i.e., $H_n$ is a graph whose vertex set is $\mathbb{F}_2^n$ and two vertices $\mathbf{u}$ and $\mathbf{v}$
are connected by an edge if $\mathbf{v}$ can be obtained from $\mathbf{u}$ by flipping one of the bits of $\mathbf{u}$, i.e., if $\sym{wt}(\mathbf{u}\oplus \mathbf{v})=1$.
Let $A$ be a subset of the vertices of $H_n$ and $\overline{A}=\mathbb{F}_2^n\setminus A$. Let $e(A,\overline{A})$ be the number of edges between $A$ and $\overline{A}$.
Suppose $f$ is an $n$-variable Boolean function such that $\sym{supp}(f)=A$. It is known that $\sym{inf}(f)=e(A,\overline{A})/2^{n-1}$ 
(see~\cite{Ka16} and Page~52 of~\cite{o2014analysis}). This relation is called the edge expansion property of influence. In this section, we obtain
a general form of this relation for $t$-$\sym{inf}(f)$.

Suppose $\mathbf{u}$ is a vertex of $H_n$ and $\bm{\alpha}\in \mathbb{F}_2^n$ with $T=\sym{supp}(\bm{\alpha})$ and $t=\#T$. Let $\mathbf{v}=\mathbf{u}\oplus \bm{\alpha}$. 
Then $\mathbf{v}$ is obtained from $\mathbf{u}$ by flipping the bits of $\mathbf{u}$ which are indexed by $T$. Since these bits can be flipped in any order, there
are a total of $t!$ paths of length $t$ in $H_n$ between $\mathbf{u}$ and $\mathbf{v}$. 

Let $A$ be a subset of $H_n$ and $f$ be an $n$-variable Boolean function such that $\sym{supp}(f)=A$. For $\bm{\alpha}\in \mathbb{F}_2^n$, let $n_{\bm{\alpha}}$ be
the number of paths between $A$ and $\overline{A}$ such that the two ends $\mathbf{u}$ and $\mathbf{v}$ of any such path satisfy $\mathbf{u}\oplus \mathbf{v}=\bm{\alpha}$.
The following result relates $n_{\bm{\alpha}}$ to the autocorrelation of $f$ at $\bm{\alpha}$.
\begin{proposition}\label{prop-ac-path}
	$\displaystyle C_f(\bm{\alpha}) = 1 - \frac{n_{\bm{\alpha}}}{(\sym{wt}(\bm{\alpha}))!2^{n-2}}$.
\end{proposition}
\begin{proof}
	Let 
	$x_{\bm{\alpha}}=\#\{(\mathbf{u},\mathbf{v}):\mathbf{u}\in A,\ \mathbf{v}\in\overline{A},\ \mathbf{u}\oplus \mathbf{v}=\bm{\alpha}\}.$
	Then 
	\begin{eqnarray}\label{eqn-x-n}
		n_{\bm{\alpha}} & = & (\sym{wt}(\bm{\alpha}))! x_{\bm{\alpha}}.
	\end{eqnarray}
	Note that 
	$x_{\bm{\alpha}}=\#\{\mathbf{u}\in\mathbb{F}_2^n: f(\mathbf{u})=1 \mbox{ and } f(\mathbf{u}\oplus \bm{\alpha})=0\}.$
	Let $g(\mathbf{X})=f(\mathbf{X})\oplus f(\mathbf{X}\oplus \bm{\alpha})$. Then 
	\begin{eqnarray}
		\sym{wt}(g)
		& = & \#\{\mathbf{u}\in \mathbb{F}_2^n: \mbox{ either } f(\mathbf{u})=1 \mbox{ and } f(\mathbf{u}\oplus \bm{\alpha})=0,
		\mbox{ or } 
		f(\mathbf{u})=0 \mbox{ and } f(\mathbf{u}\oplus \bm{\alpha})=1\} \nonumber \\
		& = & 2 \#\{\mathbf{u}\in\mathbb{F}_2^n: f(\mathbf{u})=1 \mbox{ and } f(\mathbf{u}\oplus \bm{\alpha})=0\} \nonumber \\
		& = & 2x_{\bm{\alpha}}. \label{eqn-x-alpha}
	\end{eqnarray}
	From the definition of $C_f(\bm{\alpha})$ given in~\eqref{eqn-ac-wt}, it follows that $\sym{wt}(g)=2^{n-1}(1-C_f(\bm{\alpha}))$ which combined 
	with~\eqref{eqn-x-n} and~\eqref{eqn-x-alpha} shows the result.
\end{proof}

\begin{remark}\label{rem-geom}
	Proposition~\ref{prop-ac-path} connects auto-correlation to number of paths and consequently provides a geometric interpretation of the auto-correlation
	function. Combining Proposition~\ref{prop-ac-path} with~\eqref{eqn-f-ac}, we obtain
	\begin{eqnarray*}
\left(W_{f}(\bm{\beta})\right)^2 
		& = & \Delta_{\bm{\beta}} - \frac{1}{2^{2n-2}} \sum_{\bm{\alpha}\in\mathbb{F}_2^n}
				(-1)^{\langle\bm{\alpha},\bm{\beta}\rangle} \frac{n_{\bm{\alpha}}}{(\sym{wt}(\bm{\alpha}))!},
	\end{eqnarray*}
	where $\Delta_{\bm{\beta}}=1$ if $\bm{\beta}=\mathbf{0}_n$ and $0$ otherwise. This provides a geometric interpretation of the Walsh transform. 
	To the best of our knowledge, these geometric interpretations of the auto-correlation function and the Walsh transform do not appear earlier in the literature.
\end{remark}

Now we are ready to state the path expansion property of $t$-$\sym{inf}(f)$. 
\begin{theorem}\label{thm-path-expansion}
	Let $f$ be an $n$-variable Boolean function and $t\in [n]$. Then
	\begin{eqnarray}
		t\mbox{-}\sym{inf}(f) 
		& = & 1-\frac{1}{2^{n+t-2}{n\choose t}} \sum_{\bm{\alpha}\in\mathbb{F}_2^n}{n-\sym{wt}(\bm{\alpha})\choose t-\sym{wt}(\bm{\alpha})}
		\left(2^{n-2} - \frac{n_{\bm{\alpha}}}{(\sym{wt}(\bm{\alpha}))!} \right). \label{eqn-path-expansion}
	\end{eqnarray}
\end{theorem}
\begin{proof}
	Using Proposition~\ref{prop-ac-path} in the definition of $\sym{inf}_T(f)$ given by~\eqref{eqn-inf-set}, we have
	\begin{eqnarray}
		\sym{inf}_T(f) 
		& = & 1 - \frac{1}{2^t}\left(\sum_{\bm{\alpha}\leq \chi_T} C_f(\bm{\alpha}) \right) \nonumber \\
		& = & 1 - \frac{1}{2^t}\sum_{k=0}^t\left( \sum_{\bm{\alpha}\leq \chi_T, \sym{wt}(\bm{\alpha})=k} C_f(\bm{\alpha}) \right) \nonumber \\
		& = & 1 - \frac{1}{2^t}\sum_{k=0}^t\left( \sum_{\bm{\alpha}\leq \chi_T, \sym{wt}(\bm{\alpha})=k} \left(1 - \frac{n_{\bm{\alpha}}}{k!2^{n-2}} \right) \right).
		\label{eqn-inf_T-n}
	\end{eqnarray}
	For $\bm{\alpha}\in\mathbb{F}_2^n$ with $\sym{wt}(\bm{\alpha})=k$, there are exactly ${n-k\choose t-k}$ subsets $T$ of $[n]$ such that
	$\alpha\leq \chi_T$. Using this observation, we have
	\begin{eqnarray}
		t\mbox{-}\sym{inf}(f)
		& = & \frac{1}{{n\choose t}} \sum_{T\subseteq [n], \#T=t} \sym{inf}_f(T) \nonumber \\
		& = & 1-\frac{1}{2^t{n\choose t}}\sum_{k=0}^t\left( \sum_{\{\bm{\alpha}: \sym{wt}(\bm{\alpha})=k\}} {n-k\choose t-k} 
		\left(1-\frac{n_{\bm{\alpha}}}{k!2^{n-2}} \right) \right) \nonumber \\
		& = & 1-\frac{1}{2^{n+t-2}{n\choose t}} \sum_{\bm{\alpha}\in\mathbb{F}_2^n}{n-\sym{wt}(\bm{\alpha})\choose t-\sym{wt}(\bm{\alpha})}
		\left(2^{n-2} - \frac{n_{\bm{\alpha}}}{(\sym{wt}(\bm{\alpha}))!} \right). \nonumber
	\end{eqnarray}
\end{proof}
Putting $t=1$ in~\eqref{eqn-path-expansion}, we obtain $1$-$\sym{inf}(f)=\sum_{i\in[n]}n_{\mathbf{e}_i}/(n2^{n-1})=e(A,\overline{A})/(n2^{n-1})$ which 
is the previously mentioned edge expansion property for $\sym{inf}(f)$ scaled by a factor of $n$.

\subsection{Probabilistic Interpretation \label{subsec-prob} }
We have defined the influence of a set of variables using the auto-correlation function. In this section, we provide two probabilistic interpretations of
the influence.

Let $f$ be an $n$-variable Boolean function and $\emptyset \neq T\subseteq [n]$, with $\#T=t$. We define the following probability
\begin{eqnarray}
	\mu_f(T) & = & \displaystyle \Pr_{\bm{\alpha}\leq\chi_T, \mathbf{u}\in\mathbb{F}_2^n}[f(\mathbf{u}) \neq f(\mathbf{u}\oplus \bm{\alpha})]. \label{eqn-mu} 
\end{eqnarray}
In~\eqref{eqn-mu}, $\bm{\alpha}$ is required to be chosen uniformly at random from the set $\{\mathbf{x}: \mathbf{x}\leq\chi_T\}$. This is achieved by fixing
the positions of $\bm{\alpha}$ corresponding to the elements of $\overline{T}$ to be 0, and choosing the bits of $\bm{\alpha}$ corresponding to the positions
in $T$ uniformly at random.

The definition of influence given by Fischer et al.~\cite{fischer02} and Blais~\cite{blais2009testing} is $I_f(T)$ and is given by~\eqref{eqn-blais}.
This definition is made in terms of the function $Z(T,\mathbf{x},\mathbf{y})$.
For $\mathbf{x},\mathbf{y}\in \mathbb{F}_2^n$, both $\mathbf{x}$ and $Z(T,\mathbf{x},\mathbf{y})$ agree on the bits indexed by $\overline{T}$. In particular,
the bits of $\mathbf{y}$ indexed by $\overline{T}$ do not play any role in the probability 
$\Pr_{\mathbf{x},\mathbf{y}\in \mathbb{F}_2^n} \left[ f(\mathbf{x}) \neq f(Z(T,\mathbf{x},\mathbf{y})) \right]$. So this probability
is the same as the probability of the event arising from choosing $\bm{\beta}$ uniformly at random from $\mathbb{F}_2^{n-t}$, choosing $\mathbf{w}$ and $\mathbf{z}$ 
independently and uniformly from $\mathbb{F}_2^t$ and considering $f_{\bm{\beta}}(\mathbf{w})\neq f_{\bm{\beta}}(\mathbf{z})$. This shows that 
\begin{eqnarray}
I_f(T) & = & \displaystyle \Pr_{\bm{\beta}\in\mathbb{F}_2^{n-t}, \mathbf{w},\mathbf{z}\in\mathbb{F}_2^t}[f_{\bm{\beta}}(\mathbf{w}) \neq f_{\bm{\beta}}(\mathbf{z})]. \label{eqn-nu} 
\end{eqnarray}
where $f_{\bm{\beta}}$ denotes $f_{\mathbf{X}_{\overline{T}}\leftarrow \bm{\beta}}$.

The following result relates the above two probabilities to influence.
\begin{theorem}\label{thm-inf-prob}
Let $f$ be an $n$-variable Boolean function and $\emptyset \neq T\subseteq [n]$. Then $\mu_f(T)=I_f(T)=\sym{inf}_f(T)/2$.
\end{theorem}
\begin{proof}
	We separately show that $\mu_f(T)=\sym{inf}_f(T)/2$ and $I_f(T)=\sym{inf}_f(T)/2$.
	Let $t=\#T$. 
	\begin{eqnarray}
		\mu_f(T) & = & \frac{1}{2^t} \sum_{\bm{\alpha}\leq \chi_T}\Pr_{\mathbf{u}\in\mathbb{F}_2^n} [f(\mathbf{u}) \neq f(\mathbf{u}\oplus \bm{\alpha})] \nonumber \\
		& = & \frac{1}{2^t} \sum_{\bm{\alpha}\leq \chi_T} \frac{1-C_f(\alpha)}{2} \quad \mbox{ (using~\eqref{eqn-ac-wt})} \nonumber \\
		& = & \frac{1}{2}\left(1-\frac{1}{2^t}\sum_{\bm{\alpha}\leq \chi_T}C_f(\alpha) \right) \nonumber \\
		& = & \frac{\sym{inf}_f(T)}{2}. 
	\end{eqnarray}

	\begin{eqnarray}
		I_f(T) & = & \frac{1}{2^{n-t}} 
		\sum_{\bm{\beta}\in\mathbb{F}_2^{n-t}} \Pr_{\mathbf{w},\mathbf{z}\in\mathbb{F}_2^t} [f_{\bm{\beta}}(\mathbf{w}) \neq f_{\bm{\beta}}(\mathbf{z})] \nonumber \\
		& = & \frac{1}{2^{n-t}} 
		\sum_{\bm{\beta}\in\mathbb{F}_2^{n-t}} 2 \times \frac{\sym{wt}(f_{\bm{\beta}})}{2^t}\left(1 - \frac{\sym{wt}(f_{\bm{\beta}})}{2^t} \right) \nonumber \\
		& = & \frac{1}{2^{n-1-t}} 
		\sum_{\bm{\beta}\in\mathbb{F}_2^{n-t}}  \mathbb{E}(f_{\bm{\beta}}) (1-\mathbb{E}(f_{\bm{\beta}})) \nonumber \\
		& = & \frac{1}{2^{n-1-t}} \sum_{\bm{\beta}\in\mathbb{F}_2^{n-t}}  \sym{Var}(f_{\bm{\beta}}) \nonumber \\
		& = & \frac{\sym{inf}_f(T)}{2} \quad \mbox{ (from~\eqref{eqn-inf-ac-exp-var}).} \nonumber
	\end{eqnarray}
\end{proof}
Using the third point of Theorem~\ref{thm-inf-basic}, a consequence of Theorem~\ref{thm-inf-prob} is that both the probabilities $\mu_f(T)$ and $I_f(T)$ are at most $1/2$.

\begin{remark} \label{rem-blais}
	From Theorem~\ref{thm-inf-prob}, it follows that $I_f(T)=\sym{inf}_f(T)/2$.
	Some of the results for $\sym{inf}_T(f)$ that we have proved have been obtained for $I_f(T)$ in~\cite{fischer02,blais2009testing}. In particular, it 
	has been shown that $I_f(T)$ is equal to half the right hand side of~\eqref{eqn-inf-fourier} using a somewhat long proof which is different from the one that we given. 
	Since we defined influence using the auto-correlation function, we were able to use known results on Walsh transform which make our proof simpler.
	Further, it has been proved in~\cite{fischer02,blais2009testing} that $I_{f}(T)\leq I_{f}(S\cup T)\leq I_{f}(S) + I_{f}(T)$, i.e.,
monotonicity and sub-additivity properties hold for $I_{f}$. These properties for $\sym{inf}_f(T)$ are covered by Points~4 and~5 of Theorem~\ref{thm-inf-basic}. 
\end{remark}

\subsection{Juntas \label{subsec-junta} }
The total influence of the individual variable, i.e. $\sym{inf}(f)$, for an $s$-junta $f$ is known to be at most $s$. 
The following result generalises this to provide an upper bound on $t$-$\sym{inf}(f)$ for an $s$-junta.

\begin{proposition}\label{prop-junta-upp-bnd}
	Let $f$ be an $n$-variable function which is an $s$-junta for some $s\in [n]$. For $t\in [n]$,
	$t\mbox{-}\sym{inf}(f) \leq 1 - {n-s\choose t}/{n\choose t}$.
\end{proposition}
\begin{proof}
	Let $T\subseteq [n]$ with $\#T=t$. 
	Since $f$ is an $s$-junta, there is a subset $S\subseteq [n]$, with $\#S\leq s$ such that $f$ is degenerate on the variables indexed by $\overline{S}$. 
	So $\sym{inf}_f(T)=0$ if $T$ is a subset of $\overline{S}$. This means that for ${n-s\choose t}$ possible subsets $T$, $\sym{inf}_f(T)=0$. For the other
	${n\choose t}-{n-s\choose t}$ possible subsets $T$, $\sym{inf}_f(T)\leq 1$. The result now follows from the definition of $t\mbox{-}\sym{inf}(f)$ 
	given in~\eqref{eqn-t-inf}.
\end{proof}
For $t=1$, the upper bound on $1\mbox{-}\sym{inf}(f)$ given by Proposition~\ref{prop-junta-upp-bnd} is $s/n$ which is a scaled version of the bound
$\sym{inf}(f)\leq s$. Note that the upper bound on $t\mbox{-}\sym{inf}(f)$ increases as $t$ increases and reaches 1 for $t>n-s$.

An $n$-variable Boolean function $f$ is said to be $\epsilon$-far from being a $s$-junta if for every $n$-variable $s$-junta $g$,
$\Pr_{\mathbf{x}\in \mathbb{F}_2^n}[f(\mathbf{x})\neq g(\mathbf{x})] \geq \epsilon$. 
It was proved in~\cite{blais2009testing} that if $f$ is $\epsilon$-far from being an $s$-junta, then for any set $S\subseteq [n]$ with $\#S\leq s$, $I_f(\overline{S}) \geq \epsilon$.
The following result provides an equivalent statement for $\sym{inf}_f(\overline{S})$. The reason for stating the result in the present work is that our proof is simpler than
that in~\cite{blais2009testing}. 
\begin{proposition} \label{prop-junta}
	If an $n$-variable Boolean function $f$ is $\epsilon$-far from being an $s$-junta, then for any set $S\subseteq [n]$ with $\#S\leq s$,
	$\sym{inf}_f(\overline{S}) \geq 2\epsilon$.
\end{proposition}
\begin{proof}
	Among all the $s$-juntas on the variables indexed by $S$, let $g$ be the closest $s$-junta to $f$. For $\bm{\alpha}\in \mathbb{F}_2^s$, let 
$f_{\bm{\alpha}}=f_{\mathbf{X}_S\leftarrow \bm{\alpha}}(\mathbf{X}_{\overline{S}})$ and 
	$g_{\bm{\alpha}}=g_{\mathbf{X}_S\leftarrow \bm{\alpha}}(\mathbf{X}_{\overline{S}})$ be functions on $(n-s)$-variables. Since $g$ is a junta on $S$, it is degenerate 
	on all variables indexed
	by $\overline{S}$. So $g_{\bm{\alpha}}$ is a constant function for all $\bm{\alpha}\in\mathbb{F}_2^s$. Since among all the juntas on the variables indexed by $S$, 
	$g$ is the closest $s$-junta to $f$, it follows that for each $\bm{\alpha}\in\mathbb{F}_2^s$, $g_{\bm{\alpha}}$ is either the constant function 0 or the
	constant function 1 according as $\sym{wt}(f_{\bm{\alpha}})\leq 2^{n-s-1}$ (i.e. $\mathbb{E}(f_{\bm{\alpha}})\leq 1/2$) or 
	$\sym{wt}(f_{\bm{\alpha}})> 2^{n-s-1}$ (i.e. $\mathbb{E}(f_{\bm{\alpha}})> 1/2$) respectively. So
	\begin{eqnarray}
		\Pr_{\mathbf{x}\in\mathbb{F}_2^n} [f(\mathbf{x})\neq g(\mathbf{x})]
		& = & \frac{ \sum_{\bm{\alpha}\in \mathbb{F}_2^s} \sym{wt}(f_{\bm{\alpha}} \oplus g_{\bm{\alpha}})}{2^{n}} \nonumber \\
		& = & \frac{1}{2^s}\sum_{\bm{\alpha}\in \mathbb{F}_2^s} \min\left( \mathbb{E}(f_{\bm{\alpha}}), 1- \mathbb{E}(f_{\bm{\alpha}})\right). \label{eqn-junta}
	\end{eqnarray}
Since $f$ is $\epsilon$-far from being an $s$-junta, it follows that $\epsilon\leq \Pr_{\mathbf{x}\in\mathbb{F}_2^n} [f(\mathbf{x})\neq g(\mathbf{x})]$. 
	Using $\sym{Var}(f_{\bm{\alpha}})=\mathbb{E}(f_{\bm{\alpha}})(1-\mathbb{E}(f_{\bm{\alpha}}))$, 
	it is easy to check that $\min\left( \mathbb{E}(f_{\bm{\alpha}}), 1- \mathbb{E}(f_{\bm{\alpha}})\right) \leq 2\sym{Var}(f_{\bm{\alpha}}).$ 
	The result now follows by taking $T=\overline{S}$ in~\eqref{eqn-inf-ac-exp-var} and combining with~\eqref{eqn-junta}.
\end{proof}

\subsection{Cryptographic Properties \label{subsec-crypto}}
An $n$-variable Boolean function $f$ is $\delta$-close to an $s$-junta if there is an $s$-junta $g$ such that 
$\Pr_{\mathbf{x}\in\mathbb{F}_2^n}[f(\mathbf{x})\neq g(\mathbf{x})]\leq \delta$. From the point of view of cryptographic design, it is undesirable for $f$ to be
$\delta$-close to an $s$-junta for $\delta$ close to $0$ and $s$ smaller than $n$. Since otherwise, $g$ is a good approximation of $f$ and a cryptanalyst may replace
$f$ by $g$ which may help in attacking a cipher which uses $f$ as a building block. For example, in linear cryptanalysis the goal is to obtain $g$ to be a linear 
function on a few variables such that it is a good approximation of $f$. To defend against such attacks, one usually requires $f$ to not have any good linear
approximation on a small number of variables. In particular, an $m$-resilient function cannot be approximated with probability different from 1/2 by any linear function 
on $m$ or smaller number of variables. 
A characterisation of resilient functions in terms of influence is given by Theorem~\ref{thm-max-t-inf} which shows that an $n$-variable function is
$m$-resilient if and only if $(n-m)$-$\sym{inf}(f)$ takes its maximum value of $1$. 

The next result provides a characterisation of bent functions in terms of influence.
\begin{theorem}\label{thm-bent-inf}
	Let $f$ be an $n$-variable Boolean function. Then $f$ is bent if and only if for any non-empty $T\subseteq [n]$, $\sym{inf}_f(T)=1-2^{\#T}$.
\end{theorem}
\begin{proof}
	First suppose that $f$ is bent. So $W_f(\bm{\alpha})=\pm 2^{-n/2}$ for all $\bm{\alpha}\in \mathbb{F}_2^n$. From~\eqref{eqn-ac-wf}, it follows that
	$C_f(\mathbf{x})=0$ for all $\mathbf{0}_n\neq\mathbf{x}\in\mathbb{F}_2^n$. Consequently, from~\eqref{eqn-inf-set} we have that for any non-empty 
	$T\subseteq [n]$, $\sym{inf}_f(T)=1-2^{\#T}$.

	Next we prove the converse.
	From~\eqref{eqn-inf-set}, it follows that $\sym{inf}_f(T)=1-2^{\#T}$ if and only if 
	\begin{eqnarray}\label{eqn-bent-ac}	
		\sum_{\mathbf{0}_n\neq \bm{\alpha}\leq \chi_T} C_f(\bm{\alpha}) & = & 0.
	\end{eqnarray}
	For $0\leq i\leq 2^n-1$, let $\sym{bin}_n(i)$ denote the $n$-bit binary representation of $i$. 
	Let $\mathbf{M}$ be the $(2^n-1)\times (2^n-1)$ matrix whose rows and columns are indexed by the integers in $[2^n-1]$ such that the $(i,j)$-th
	entry of $\mathbf{M}$ is 1 if $\sym{bin}_n(j)\leq \sym{bin}_n(i)$ and otherwise the entry is 0. It is easy to verify that $\mathbf{M}$ is a lower triangular
	matrix whose diagonal elements are all 1. In particular, $\mathbf{M}$ is invertible. 

	Let $\mathbf{C}=[C_f(\sym{bin}_n(i))]_{i\in [2^n-1]}$ be the vector of auto-correlations of $f$ at all the non-zero points in $\mathbb{F}_2^n$. The set of 
	relations of the form~\eqref{eqn-bent-ac} for all non-empty $T\subseteq [n]$ can be expressed as $\mathbf{M}\mathbf{C}^{\top}=\mathbf{0}^{\top}$. 
	Since $\mathbf{M}$ is invertible, it follows that $\mathbf{C}=\mathbf{0}$, i.e. $C_f(\bm{\alpha})=0$ for all non-zero $\bm{\alpha}\in \mathbb{F}_2^n$. 
	From~\eqref{eqn-f-ac}, it now follows that $W_f(\bm{\beta})=\pm 2^{-n/2}$ for all $\bm{\beta}\in \mathbb{F}_2^n$ which shows that $f$ is bent.
\end{proof}

For functions satisfying propagation characteristics, somewhat less can be said. 
From~\eqref{eqn-inf-set}, it follows that if $f$ satisfies PC($k$) then for any subset $\emptyset\neq T\subseteq [n]$ with $\#T=t\leq k$, $\sym{inf}_f(T)=1-2^{-t}$ and so
$t$-$\sym{inf}(f)=1-2^{-t}$. 


\subsection{The Fourier Entropy/Influence Conjecture\label{subsec-FEI}}
The Fourier entropy $H(f)$ of $f$ is defined to be the entropy of the probability distribution $\{W_{f}^2(\bm{\alpha})\}$ and is equal to
\begin{eqnarray}\label{eqn-H}
	H(f) & = & -\sum_{\bm{\alpha}\in \mathbb{F}_2^n}W_{f}^2(\bm{\alpha}) \log W_{f}^2(\bm{\alpha}),
\end{eqnarray}
where $\log$ denotes $\log_2$ and the expressions $0\log 0$ and $0\log \frac{1}{0}$ are to be interpreted as $0$. For $t\in [n]$, let
\begin{eqnarray}\label{eqn-rho}
	\rho_t(f) & = & \frac{H(f)/n}{t\mbox{-}\sym{inf}(f)}.
\end{eqnarray}
The Fourier entropy/influence conjecture~\cite{friedgut1996every} states that there is a universal constant $C$, such that for all Boolean functions $f$, $\rho_1(f)\leq C$. 
A general form of this conjecture is that there is a universal constant $C_t$, such that for all Boolean functions $f$ and $t\in [1,n]$, $\rho_t(f)\leq C_t$. 
Since $t$-$\sym{inf}(f)$ increases monotonically with $t$, it follows that $\rho_t(f)$ decreases monotonically with $t$. 
So if the FEI conjecture holds, then the conjecture on $\rho_t(f)$ also holds for $t\geq 1$. The converse, i.e if the conjecture holds for some $\rho_t$ with $t>1$ then it
also holds for $\rho_1$, need not be true. 

\begin{remark}\label{rem-FEI-variant}
	A weaker variant of the FEI conjecture replaces $H(f)$ by the min-entropy of the distribution $\widehat{p}_f(\omega)$. In a similar
	vein, one may consider the conjecture on $\rho_t(f)$ to be a weaker variant of the FEI conjecture. 
\end{remark}

\section{Pseudo-Influence \label{sec-PI} }
In this section, we define a quantity based on the auto-correlation function which we call the pseduo-influence of a Boolean function. The main reason for 
considering this notion is that it
turns out to be the same as the notion of influence $J_f(T)$ introduced in~\cite{tal2017tight}. We make a thorough study of the basic properties of pseudo-influence.
A consequence of this study is that pseudo-influence
does not satisfy some of the basic desiderata that a notion of influence may be expected to satisfy, which is why we call it pseudo-influence. This shows that even though 
the quantity was termed `influence' in~\cite{tal2017tight}, it is not a satisfactory notion of influence.

Suppose $f(\mathbf{X})$ is an $n$-variable Boolean function where $\mathbf{X}=(X_1,\ldots,X_n)$ and $\emptyset \neq T=\{i_1,\ldots,i_t\}\subseteq [n]$. 
We define pseudo-influence $\sym{PI}_f(T)$ of the set of variables $\{X_{i_1},\ldots,X_{i_t}\}$ indexed by $T$ on $f$ in the following manner.

\begin{eqnarray}\label{eqn-inf-set-PI}
	\sym{PI}_f(T) & = & \frac{1}{2^{\#T}}\left(\sum_{\bm{\alpha}\leq \chi_T} (-1)^{\sym{wt}(\alpha)}C_f(\bm{\alpha})\right).
\end{eqnarray}
For a singleton set $T=\{i\}$, $\sym{PI}_f(T)=\sym{inf}_f(T)=\sym{inf}_f(i)$. 

Let $f$ be an $n$-variable function and $t$ be an integer with $1\leq t\leq n$. Then the $t$-pseudo-influence of $f$ is the total pseudo-influence (scaled by ${n\choose t}$) 
obtained by summing the pseudo-influence of every set of $t$ variables on the function $f$, i.e., 
\begin{eqnarray}\label{eqn-t-PI}
	t\mbox{-}\sym{PI}(f) & = & \frac{\sum_{\{T\subseteq [n]: \#T=t\}} \sym{PI}_f(T)}{{n\choose t}}.
\end{eqnarray}

The characterisation of pseudo-influence in terms of the Walsh transform is given by the following result.
\begin{theorem}\label{thm-inf-fourier-PI}
	Let $f$ be an $n$-variable Boolean function and $\emptyset \neq T\subseteq [n]$. Then
	\begin{eqnarray}\label{eqn-inf-fourier-PI}
		\sym{PI}_f(T) & = & \sum_{\mathbf{u} \geq \chi_T} \left(W_{f}(\mathbf{u})\right)^2.
	\end{eqnarray}
	Consequently, for an integer $t$ with $1\leq t\leq n$,
	\begin{eqnarray}
		t\mbox{-}\sym{PI}(f) & = & \frac{1}{{n\choose t}}\sum_{k=t}^n {k\choose t} \widehat{p}_f(k) \label{eqn-PI-fourier-t}
	\end{eqnarray}
\end{theorem}
\begin{proof}
    Let $\#T=t$. Let $E=\{\bm{\beta}\in\mathbb{F}_2^n: \bm{\beta}\leq \chi_{\overline{T}}\}$. Then $\#E=2^{n-t}$ and 
	$E^{\perp}=\{\bm{\alpha}\in\mathbb{F}_2^n: \bm{\alpha}\leq \chi_T\}$. 
	From~\eqref{eqn-inf-set-PI} and putting $\mathbf{a}=\mathbf{1}_n$, $\mathbf{b}=\mathbf{0}_n$ and $\psi=C_f$ 
	in~\eqref{poisson-first-order} we obtain the following:
    \begin{eqnarray*}
	    \sym{PI}_f(T)
	    & = & \frac{1}{2^t}\sum_{\bm{\alpha}\leq \chi_{T}}(-1)^{\sym{wt}(\bm{\alpha)}}C_f(\bm{\alpha})
	    = \frac{1}{2^t}\sum_{\bm{\alpha}\leq \chi_{T}}(-1)^{\langle\mathbf{1}_n,\bm{\alpha\rangle}}C_f(\bm{\alpha})
	    =\sum_{\bm{\beta}\in\mathbf{1}_n+E}\widehat{C_f}(\bm{\beta})
	    =\sum_{\bm{\beta}\geq \chi_{T}}\widehat{C_f}(\bm{\beta}).
	\end{eqnarray*}
	The result now follows from~\eqref{eqn-f-ac}. 

	The expression for $t\mbox{-}\sym{PI}(f)$ can be seen as follows.
	\begin{eqnarray}
		t\mbox{-}\sym{PI}(f) 
		& = & \frac{1}{{n\choose t}}\sum_{k=t}^n\sum_{\{\mathbf{u}:\sym{wt}(\mathbf{u})=k\}} {k\choose t} \left( W_{f}(\mathbf{u}) \right)^2 \nonumber \\
		& = & \frac{1}{{n\choose t}}\sum_{k=t}^n {k\choose t} \sum_{\{\mathbf{u}:\sym{wt}(\mathbf{u})=k\}} \left( W_{f}(\mathbf{u}) \right)^2 \nonumber \\
		& = & \frac{1}{{n\choose t}}\sum_{k=t}^n {k\choose t} \widehat{p}_f(k) \nonumber \\
		& = & \frac{1}{{n\choose t}}\sum_{k=t}^n {k\choose t} \widehat{p}_f(k). 
	\end{eqnarray}
\end{proof}

The following result states the basic properties of the pseudo-influence.

\begin{theorem} \label{thm-inf-basic-PI}
        Let $f$ be an $n$-variable Boolean function and $\emptyset \neq T \subseteq S \subseteq [n]$. Then
	\begin{enumerate}
		\item $0\leq \sym{PI}_f(T) \leq 1$. 
		\item If the function $f$ is degenerate on the variables indexed by $T$, then $\sym{PI}_f(T)=0$.
		\item $\sym{PI}_f(S) \leq \sym{PI}_f(T)$. 
	\end{enumerate}
\end{theorem}
\begin{proof}
	The first point follows from Theorem~\ref{thm-inf-fourier-PI} and Parseval's theorem. The third point also follows from Theorem~\ref{thm-inf-fourier}.

	Consider the second point. Suppose $\pi$ is any permutation of $[n]$ and define $g(\mathbf{X})$ to be the function 
	$f(X_{\pi(1)},\ldots,X_{\pi(n)})$. Then $f$ is degenerate on the variables indexed by a set $U=\{i_1,\ldots,i_t\}$ if and only if $g$ is degenerate on the variables
	indexed by the set $V=\{\pi(i_1),\ldots,\pi(i_t)\}$. Also, $\sym{inf}_f(U)=\sym{inf}_g(V)$. 
	In view of this, we consider the set $T$ to be $\{1,\ldots,t\}$.

	For $\bm{\alpha}\in\mathbb{F}_2^t$ and $\mathbf{Y}=(X_{t+1},\ldots,X_n)$, let $f_{\bm{\alpha}}(\mathbf{Y})=f(\bm{\alpha},\mathbf{Y})$. The function $f$
	is degenerate on the variables indexed by $T$ if and only if $f_{\bm{\alpha}}(\mathbf{Y})=f_{\bm{\beta}}(\mathbf{Y})$ for any 
	$\bm{\alpha},\bm{\beta}\in \mathbb{F}_2^t$. We show that the latter condition is equivalent to $f(\mathbf{X})=f(\mathbf{X}\oplus \bm{\gamma})$ for 
	any $\bm{\gamma}\leq \chi_{T}$. 
	Note that by the choice of $T$, we have that for $\bm{\gamma}\leq \chi_T$, $\bm{\gamma}=(\bm{\delta},\mathbf{0})$ for some $\bm{\delta}\in \mathbb{F}_2^t$. 
	So it is sufficient to show that $f(\bm{\alpha},\mathbf{Y})=f((\bm{\alpha},\mathbf{Y})\oplus (\bm{\delta},\mathbf{0}))$ for all $\bm{\alpha}\in \mathbb{F}_2^t$.
	The latter condition is equivalent to $f_{\bm{\alpha}}(\mathbf{Y}) = f_{\bm{\alpha}\oplus\bm{\delta}}(\mathbf{Y})=f_{\bm{\beta}}(\mathbf{Y})$ where 
	$\bm{\beta}=\bm{\alpha}\oplus\bm{\delta}$. This completes the proof that $f$ is degenerate on the variables indexed by $T$ if and only if
	$f(\mathbf{X})=f(\mathbf{X}\oplus \bm{\gamma})$ for all $\bm{\gamma}\leq \chi_{T}$.

	The condition $f(\mathbf{X})=f(\mathbf{X}\oplus \bm{\gamma})$ for all $\bm{\gamma}\leq \chi_{T}$ is equivalent to $C_f(\bm{\gamma})=1$ for all
	$\bm{\gamma}\leq \chi_T$. So $f$ is degenerate on the set of variables indexed by $T$ if and only if $C_f(\bm{\gamma})=1$ for all
        $\bm{\gamma}\leq \chi_T$. Using this in the definition of pseudo-influence given by~\eqref{eqn-inf-set-PI}, we obtain the
	the second point. 
\end{proof}

Theorem~\ref{thm-inf-basic-PI} states that if $f$ is degenerate on the variables indexed by $T$, then $\sym{PI}_f(T)=0$. The converse, however, is not true. 
Suppose $f$ is an $n$-variable function such that $W_{f}(\mathbf{1}_n)=0$ and let $T=[n]$. Then from~\eqref{eqn-inf-fourier-PI}, $\sym{PI}_f(T)=0$.
This example can be generalised. Suppose $g$ is an $n$-variable, $m$-resilient function and let $f(\mathbf{X})=\langle \mathbf{1},\mathbf{X}\rangle \oplus g(\mathbf{X})$.
Using~\eqref{eqn-wt-fou}, we have $W_{f}(\bm{\alpha})=W_{g}(\mathbf{1}\oplus\bm{\alpha})$ for all $\bm{\alpha}\in \mathbb{F}_2^n$. Since, $g$ is $m$-resilient,
$W_{g}(\bm{\omega})=0$ for all $\bm{\omega}$ with $\sym{wt}(\bm{\omega}) \leq m$. So 
$W_{f}(\bm{\alpha})=0$ for all $\bm{\alpha}$ with $\sym{wt}(\bm{\alpha}) \geq n-m$. Consequently, for any $\emptyset \neq T\subseteq [n]$, with
$\#T\geq n-m$, it follows that $\sym{PI}_f(T)=0$. There are known examples of non-degenerate resilient functions. See for example~\cite{DBLP:journals/tit/SarkarM04}.

\begin{remark}\label{rem-PI}
By the above discussion, $\sym{PI}_f(T)$ can be zero even if $f$ is non-degenerate on the variables indexed by $T$. 
Further, the third point of Theorem~\ref{thm-inf-basic-PI} shows that $\sym{PI}_f(T)$ is non-increasing with $T$. As a consequence, sub-additivity does not hold for $\sym{PI}_f(T)$. 
	So $\sym{PI}_f(T)$ violates some of the basic desiderata that one may expect a notion of influence to fulfill. 
\end{remark}

For $\mathbf{u}\in \mathbb{F}_2^n$ and $\emptyset \neq T\subseteq [n]$, $\mathbf{u}\geq \chi_T$ is equivalent to $\sym{supp}(\mathbf{u})\supseteq T$ which in particular 
implies that $\sym{supp}(u)\cap T\neq \emptyset$. So 
from~\eqref{eqn-inf-fourier} and~\eqref{eqn-inf-fourier-PI}, we have the following result which states that influence is always at least as large as the
pseudo-influence.
\begin{proposition}\label{prop-inf-gt-PI}
	Let $f$ be an $n$-variable Boolean function and $\emptyset \neq T\subseteq [n]$. Then $\sym{inf}_f(T) \geq \sym{PI}_f(T)$.
	Consequently, $t\mbox{-}\sym{inf}(f) \geq t\mbox{-}\sym{PI}(f)$ for $1\leq t\leq n$. 
\end{proposition}

\begin{theorem}\label{thm-max-t-PI}
	Let $f(\mathbf{X})$ be an $n$-variable Boolean function where $\mathbf{X}=(X_1,\ldots,X_n)$ and $t$ be an integer with $1\leq t\leq n$. 
	\begin{enumerate}
		\item $t\mbox{-}\sym{PI}(f)$ takes its maximum value of $1$ if and only if $f$ is of the form $f(\mathbf{X}) = \langle \mathbf{1},\mathbf{X}\rangle$.
		\item $t\mbox{-}\sym{PI}(f)$ takes its minimum value of $0$ if and only if $f$ is of the form 
			$f(\mathbf{X})=\langle \mathbf{1},\mathbf{X}\rangle \oplus g(\mathbf{X})$, where $g(\mathbf{X})$ is $(n-t)$-resilient.
	\end{enumerate}
\end{theorem}
\begin{proof}
    From~\eqref{eqn-PI-fourier-t}, $t\mbox{-}\sym{PI}(f)$ takes its maximum value of $1$ if and only if 
	\begin{eqnarray} \label{eqn-tt01}	
		\sum_{k=t}^n {k\choose t} \widehat{p}_f(k) & = & {n\choose t}. 
	\end{eqnarray} 
	If $f(\mathbf{X}) = \langle \mathbf{1},\mathbf{X}\rangle$, then $\widehat{p}_f(n)=1$ and $\widehat{p}_f(k)=0$ for $0\leq k\leq n-1$. 
	On the other hand, if $f(\mathbf{X}) \neq \langle \mathbf{1},\mathbf{X}\rangle$, then $\widehat{p}_f(n)<1$ and we have 
\begin{eqnarray*}
	\lefteqn{{t \choose t}\widehat{p}_f(t)+{t+1 \choose t}\widehat{p}_f(t+1) + \cdots + {n \choose t}\widehat{p}_f(n)} \\
	& \leq & {n-1\choose t} (\widehat{p}_f(0) + \cdots + \widehat{p}_f(n-1)) + {n \choose t}\widehat{p}_f(n) \\
	& = & {n-1\choose t} (1-\widehat{p}_f(n)) + {n \choose t}\widehat{p}_f(n)  <  {n \choose t}.
\end{eqnarray*}
This completes the proof of the first point.

For the second point, from~\eqref{eqn-PI-fourier-t}, one may note that the values $\widehat{p}_f(0),\ldots,\widehat{p}_f(t-1)$ do not affect the expression for 
$t\mbox{-}\sym{PI}(f)$. So $t\mbox{-}\sym{PI}(f)=0$ if and only if $\widehat{p}_f(t)=\cdots=\widehat{p}_f(n)=0$. The latter condition holds if and only if
$f$ is of the stated form.
\end{proof}
Using the second point of Theorem~\ref{thm-max-t-PI}, it is possible to obtain examples of non-degenerate functions $f$ such that $t\mbox{-}\sym{PI}(f)$ is 0.

\begin{remark} \label{rem-Tal}
The quantity $J_f(T)$ (see~\eqref{eqn-inf-tal}) was put forward by Tal~\cite{tal2017tight} as a measure of influence of the set of variables indexed by $T$
	on the function $f$. It was shown in~\cite{tal2017tight} that $J_f(T)$ is equal to the right hand side of~\eqref{eqn-inf-fourier-PI}. So it follows that
	$J_f(T)=\sym{PI}_f(T)$. This is somewhat surprising since the definition of $J_f(T)$ given in~\eqref{eqn-inf-tal} and that of $\sym{PI}_f(T)$ given 
	in~\eqref{eqn-inf-set-PI} are very different. It is perhaps only through the characterisations of both these quantities in terms of the Walsh transform
	that they can be seen to be equal. 
	The quantity $\sum_{\{T:\#T=t\}}J_{f}(T)$ was considered in~\cite{tal2017tight} and the expression~\eqref{eqn-PI-fourier-t} was also 
	obtained in~\cite{tal2017tight}. 
	Since $J_f(T)=\sym{PI}_f(T)$, from Remark~\ref{rem-PI} it follows that $J_f(T)$ is not a satisfactory notion of influence.
\end{remark}

For an $n$-variable Boolean function $f$, define $L_{1,t}=\sum_{\mathbf{u}=t}|W_f(\mathbf{u})|$ and $W^{\geq t}(f)=\sum_{i\geq t}\widehat{p}_f(i)$. Lemma~31
of~\cite{tal2017tight} showed that if for all $t$, $t\mbox{-}\sym{PI}(f)\leq C\cdot \ell^t$ for some constant $C$, then 
$W^{\geq k}(f) \leq C\cdot e\cdot\ell\cdot e^{-(k-1)/(e\ell)}$ for all $k$. Lemma~34 of~\cite{tal2017tight} showed that 
$L_{1,t}(f)\leq 2^t\cdot t\mbox{-}\sym{PI}(f)$. Since Proposition~\ref{prop-inf-gt-PI} shows that
$t\mbox{-}\sym{inf}(f) \geq t\mbox{-}\sym{PI}(f)$ for $1\leq t\leq n$, we obtain simple extensions of the Lemmas~31 and~34 of~\cite{tal2017tight} by replacing 
$t\mbox{-}\sym{PI}(f)$ with $t\mbox{-}\sym{inf}(f)$ in the above statements. Lemma~29 of~\cite{tal2017tight} provides a converse of Lemma~31. This converse does not
necessarily hold if $t\mbox{-}\sym{PI}(f)$ is replaced with $t\mbox{-}\sym{inf}(f)$. Lemmas~29 and~31 of~\cite{tal2017tight} relate spectral tail bounds to
bounds on pseudo-influence. We note that a spectral concentration result for $t\mbox{-}\sym{inf}(f)$ is given by Theorem~\ref{thm-conc}.

\section{Ben-Or and Linial Definition of Influence \label{sec-BL} }
The first notion of influence of a set of variables on a Boolean function was proposed by Ben-Or and Linial in~\cite{ben1987collective}. In this section, we introduce 
this notion, prove some of its basic properties and show its relationship with the notion of influence defined in Section~\ref{sec-inf-set}.

For an $n$-variable function $f$ and $\emptyset \neq T\subseteq [n]$, with $t=\#T$, the notion of influence introduced in~\cite{ben1987collective} is $\mathcal{I}_f(T)$
and is given by~\eqref{eqn-BL-inf}.
For $t\in [n]$, we define
\begin{eqnarray}\label{eqn-BL-inf-t}
	t\mbox{-}\mathcal{I}(f) & = & \frac{\sum_{\{T\subseteq[n]:\#T=t \}}\mathcal{I}_f(T) }{{n\choose t}}.
\end{eqnarray}

The following result provides an alternative description of $\mathcal{I}_f(T)$.
\begin{proposition} \label{prop-BL-alt}
For an $n$-variable function $f$ and $\emptyset \neq T\subseteq [n]$, with $t=\#T$,
	\begin{eqnarray} 
		\mathcal{I}_f(T) 
		& = & 1 - \frac{\#\left\{\bm{\alpha}\in \mathbb{F}_2^{n-t}: \left( W_{{f}_{\bm{\alpha}}}(\mathbf{0}_t) \right)^2 = 1 \right\}}{2^{n-t}} \label{eqn-BL-alt} \\
		& = & \frac{\#\left\{\bm{\alpha}\in \mathbb{F}_2^{n-t}: \left( W_{{f}_{\bm{\alpha}}}(\mathbf{0}_t) \right)^2 \neq 1 \right\}}{2^{n-t}}, \label{eqn-BL-alt-1}
	\end{eqnarray}
	where $f_{\bm{\alpha}}$ denotes $f_{\mathbf{X}_{\overline{T}}\leftarrow \bm{\alpha}}$. 
\end{proposition}
\begin{proof}
	From~\eqref{eqn-BL-inf}, it clearly follows that
        \begin{eqnarray*}
                \mathcal{I}_f(T)
		& = & 1 - \frac{\#\{\bm{\alpha}\in\mathbb{F}_2^{n-t}: f_{\bm{\alpha}} \mbox{ is constant} \}}{2^{n-t}} \\
		& = & 1 - \frac{\#\{\bm{\alpha}\in\mathbb{F}_2^{n-t}: \sym{wt}(f_{\bm{\alpha}})=0, \mbox{ or } 2^t \}}{2^{n-t}} \\
		& = & 1 - \frac{\#\{\bm{\alpha}\in\mathbb{F}_2^{n-t}: W_{{f}_{\bm{\alpha}}}(\mathbf{0}_t) = \pm 1 \}}{2^{n-t}}.
        \end{eqnarray*}
	This shows~\eqref{eqn-BL-alt}, and~\eqref{eqn-BL-alt-1} follows directly from~\eqref{eqn-BL-alt}.
\end{proof}

Some basic properties of $\mathcal{I}_f(T)$ are as follows.
\begin{theorem} \label{thm-basic-BL}
Let $f$ be an $n$-variable function and $\emptyset \neq T\subseteq S \subseteq [n]$. Let $\#T=t$. 
	\begin{enumerate}
		\item $0\leq \mathcal{I}_f(T) \leq 1$.
		\item $\mathcal{I}_f(T)=0$ if and only if $f$ is degenerate on the variables indexed by $T$.
		\item $\mathcal{I}_f(T)=1$ if and only if $f_{\bm{\alpha}}$ is a non-constant function for every $\bm{\alpha}\in\mathbb{F}_2^{n-t}$,
		where $f_{\bm{\alpha}}$ denotes $f_{\mathbf{X}_{\overline{T}}\leftarrow \bm{\alpha}}$. 
			In particular, if $T=[n]$, then $\mathcal{I}_f(T)=1$.
		\item $\mathcal{I}_f(T) \leq \mathcal{I}_f(S)$. 
	\end{enumerate}
\end{theorem}
\begin{proof}
The first point is obvious.

	For the second point, using~\eqref{eqn-BL-alt} note that $\mathcal{I}_f(T)=0$ if and only if for every $\bm{\alpha}\in \mathbb{F}_2^{n-t}$, 
	$W_{{f}_{\bm{\alpha}}}(\mathbf{0}_t) = \pm 1$, i.e., if and only if $\sym{wt}(f_{\bm{\alpha}})=0, \mbox{ or } 2^t$, i.e., if and only if 
$f_{\bm{\alpha}}$ is constant. The latter condition holds if and only if the variables indexed by $T$ have no
effect on the value of $f$, i.e., if and only if $f$ is degenerate on the variables indexed by $T$.

To see the third point, note that $\mathcal{I}_f(T)=1$ if and only if for every $\bm{\alpha}\in \mathbb{F}_2^{n-t}$,
	$\left(W_{{f}_{\bm{\alpha}}}(\mathbf{0}_t) \right)^2 \neq 1$, which holds if and only if $f_{\bm{\alpha}}$ is a non-constant function.

Let $\#S=s$. For the fourth point, it is sufficient to consider $s=t+1$, since otherwise, we may define a sequence of sets 
$T\subset S_1\subset S_2\subset \cdots \subset S$, with $\#T+1=\#S_1$, $\#S_1+1=\#S_2$, $\ldots,$ and argue 
	$\mathcal{I}_f(T) \leq \mathcal{I}_{f}(S_1) \leq \cdots \leq \mathcal{I}_{f}(S)$. 
	Further, without loss of generality, we assume $T=\{n-t+1,\ldots,n\}$ and $S=\{n-t,\ldots,n\}$ as otherwise, we may apply an appropriate permutation
	on the variables to ensure this condition.  Then $\overline{T}=\{1,\ldots,n-t\}$ and $\overline{S}=\{1,\ldots,n-t-1\}$. 

	Let $\mathcal{T}=\{\bm{\alpha}\in\mathbb{F}_2^{n-t}: f_{\bm{\alpha}} \mbox{ is constant}\}$
	and $\mathcal{S}=\{\bm{\beta}\in\mathbb{F}_2^{n-t-1}: f_{\bm{\beta}} \mbox{ is constant}\}$,
	where $f_{\bm{\beta}}$ is a shorthand for $f_{\mathbf{X}_{\overline{S}}\leftarrow \bm{\beta}}$. 
	Note that if $\bm{\beta}\in \mathcal{S}$, then $(\bm{\beta},0),(\bm{\beta},1)\in \mathcal{T}$. So $\#\mathcal{T}\geq 2\#\mathcal{S}$ which implies
	\begin{eqnarray*}
		\frac{\#\mathcal{T}}{2^{n-t}} \geq \frac{2\#\mathcal{S}}{2^{n-t}} \geq \frac{\#\mathcal{S}}{2^{n-t-1}}.
	\end{eqnarray*}
	Consequently, 
	\begin{eqnarray*}
		\mathcal{I}_f(T) = 1-\frac{\#\mathcal{T}}{2^{n-t}} \leq 1- \frac{\#\mathcal{S}}{2^{n-t-1}} =\mathcal{I}_f(S).
	\end{eqnarray*}
\end{proof}

\begin{remark}\label{rem-BL-sub-additivity}
We note that the sub-additivity property does not hold for $\mathcal{I}_f(T)$. As an example, consider a 6-variable function $f$ which maps
$\mathbf{0}_6$ to 1 and all other elements of $\mathbb{F}_2^6$ to 0; let $S=\{4,5,6\}$ and $T=\{2,3,6\}$. Then
$\mathcal{I}_f(S\cup T)=1/2 > 1/8 + 1/8 = \mathcal{I}_f(S) + \mathcal{I}_f(T)$. 
\end{remark}

Next, we show that the Ben-Or and Linial notion of influence is always at least as much as the notion of influence defined in~\eqref{eqn-inf-set}.
\begin{theorem} \label{thm-at-most}
	Let $f$ be an $n$-variable function and $\emptyset \neq T\subseteq [n]$. Then $\sym{inf}_f(T) \leq \mathcal{I}_f(T)$. 
	Further, equality holds if and only if $\left(W_{{f}_{\bm{\alpha}}}(\mathbf{0}_t) \right)^2 =0\mbox{ or }1$ for each $\bm{\alpha}\in \mathbb{F}_2^{n-t}$,
	where $f_{\bm{\alpha}}$ denotes $f_{\mathbf{X}_{\overline{T}}\leftarrow \bm{\alpha}}$. 
\end{theorem}
\begin{proof}
	We rewrite~\eqref{eqn-inf-ac-exp} in the following form.
	\begin{eqnarray}\label{eqn-00b}
		\sym{inf}_f(T) & = & \frac{1}{2^{n-t}} \sum_{\bm{\alpha}\in\mathbb{F}_2^{n-t}} \left( 1 - \left(W_{{f}_{\bm{\alpha}}}(\mathbf{0}_t) \right)^2 \right).
	\end{eqnarray}
	Consider the expressions for $\sym{inf}_f(T)$ and $\mathcal{I}_f(T)$ given by~\eqref{eqn-00b} and~\eqref{eqn-BL-alt-1} respectively. Both
	the expressions are sums over $\bm{\alpha}\in \mathbb{F}_2^{n-t}$. 
	Suppose $\bm{\alpha}$ is such that $\left(W_{{f}_{\bm{\alpha}}}(\mathbf{0}_t) \right)^2= 1$. The contribution of such an $\bm{\alpha}$ to both~\eqref{eqn-00b} 
	and~\eqref{eqn-BL-alt-1} is 0. Next suppose $\left(W_{{f}_{\bm{\alpha}}}(\mathbf{0}_t) \right)^2 \neq 1$; the contribution of such an $\bm{\alpha}$ 
	to~\eqref{eqn-BL-alt-1} is 1 and the contribution to~\eqref{eqn-00b} is at most 1, and the value 1 is achieved if and only if $W_{{f}_{\bm{\alpha}}}(\mathbf{0}_t)=0$.
\end{proof}

One may compare the properties of $\mathcal{I}_f(T)$ given by Theorem~\ref{thm-basic-BL} to the desiderata that a notion of influence may be expected to satisfy
(see the discussion before Theorem~\ref{thm-inf-basic}). 
The measure $\mathcal{I}_f(T)$ satisfies some of the desiderata, namely, it is between 0 and 1; takes the value 0
if and only if $f$ is degenerate on the variables indexed by $T$; and it is monotone increasing with the size of $T$. On the other hand, as noted above, it does not satisfy the
sub-additivity property. 

Compared to $\sym{inf}_f(T)$, the value of $\mathcal{I}_f(T)$ rises quite sharply. To see this, it is useful to view the following expressions for the two quantities.
\begin{eqnarray}
	2^{n-t}\times \sym{inf}_f(T) & = & \sum_{\bm{\alpha}\in\mathbb{F}_2^{n-t}}\left( 1- \left(W_{{f}_{\bm{\alpha}}}(\mathbf{0}_t)\right)^2 \right), \label{eqn-00d} \\
	2^{n-t}\times \mathcal{I}_f(T) & = & \#\left\{\bm{\alpha}\in\mathbb{F}_2^{n-t}: \left(W_{{f}_{\bm{\alpha}}}(\mathbf{0}_t) \right)^2 \neq 1 \right\}. \label{eqn-00e}
\end{eqnarray}
Suppose $\bm{\alpha}\in\mathbb{F}_2^{n-t}$ is such that $f_{\bm{\alpha}}$ is a non-constant function, so that $\left(W_{{f}_{\bm{\alpha}}}(\mathbf{0}_t) \right)^2 \neq 1$. 
Then such an $\bm{\alpha}$ contributes 1 to~\eqref{eqn-00e}, while it contributes a value which is at most 1 to~\eqref{eqn-00d}. More generally, $\bm{\alpha}$
contributes either $0$ or $1$ to~\eqref{eqn-00e} according as $f_{\bm{\alpha}}$ is constant or non-constant; on the other hand, the contribution of 
$\bm{\alpha}$ to~\eqref{eqn-00d} is more granular. Consequently, the value of $\mathcal{I}_f(T)$ rises more sharply than the value of $\sym{inf}_f(T)$. 
In particular, if $f$ and $g$ are two distinct functions such that for all $\bm{\alpha}$, both $f_{\bm{\alpha}}$ and $g_{\bm{\alpha}}$ are non-constant functions,
then both $\mathcal{I}_f(T)$ and $\mathcal{I}_g(T)$ will be necessarily be equal to 1, whereas the values of $\sym{inf}_f(T)$ and $\sym{inf}_g(T)$ are
neither necessarily 1 nor necessarily equal. In other words, the discerning power of $\mathcal{I}_f(T)$ as a measure of influence is less than that of $\sym{inf}_f(T)$,
i.e., $\mathcal{I}_f(T)$ is a more coarse measure of influence.
So while both $\sym{inf}_f(T)$ and $\mathcal{I}_f(T)$ share some intuitive basic properties expected of a definition of influence, the facts that $\mathcal{I}_f(T)$
does not satisfy sub-additivity and has less discerning power make it a less satisfactory measure of influence compared to $\sym{inf}_f(T)$.

Theorem~\ref{thm-at-most} shows that $\sym{inf}_f(T) \leq \mathcal{I}_f(T)$. The difference between $\mathcal{I}_f(T)$ and $\sym{inf}_f(T)$ can be quite large. 
For example, if we take $f(\mathbf{X})=X_1\cdots X_n$ (i.e., the Boolean AND function), then $\mathcal{I}_f([n])=1$ while 
$\sym{inf}_f([n])=1 - (1-1/2^{n-1})^2$. In other words, the influence of the set of all variables as measured by $\mathcal{I}_f$ is 1, while the influence
as measured by $\sym{inf}_f$ is close to $0$. The influence of $[n]$ on the degenerate $n$-variable constant all-zero function is $0$ as measured by 
both $\mathcal{I}_f$ and $\sym{inf}_f$. The AND function differs from the all-zero function by a single bit and so one would expect the influence of $[n]$
to remain close to 0. This is indeed the case for $\sym{inf}_f$, while for $\mathcal{I}_f$ the value jumps to 1. The example of the AND function can be 
generalised to a balanced function in the following manner. Let $1\leq t<n$ and define 
$f(X_1,X_2,\ldots,X_n)=X_1\cdots X_t\oplus X_{t+1}\oplus\cdots\oplus X_n$. It is easy to verify that $f$ is balanced. Let $T=\{1,\ldots,t\}$. One may check
that $\mathcal{I}_f(T)=1$ and $\sym{inf}_f(T)=1-(1-1/2^{t-1})^2$. As in the case of the AND function, it can be argued that one would expect the influence
of $T$ to be close to $0$ rather than being equal to $1$. 

The following result characterises the minimum and maximum values of $t\mbox{-}\mathcal{I}(f)$.
\begin{theorem}\label{thm-I-t-max-min}
	Let $f$ be an $n$-variable Boolean function and $t$ be an integer with $1\leq t\leq n$. 
	\begin{enumerate}
		\item $t$-$\mathcal{I}(f)$ takes its maximum value of 1 if and only if for every subset $T$ of $[n]$ of size $t$, and for
			every $\bm{\alpha}\in \mathbb{F}_2^{n-t}$, the function $f_{\mathbf{X}_{\overline{T}\leftarrow \bm{\alpha}}}(\mathbf{X}_T)$ is non-constant.
		\item $t$-$\mathcal{I}(f)$ takes its minimum value of 0 if and only if $f$ is a constant function.
	\end{enumerate}
\end{theorem}
\begin{proof}
	The proof of the first point follows from the third point of Theorem~\ref{thm-basic-BL}.

	For the second point, we note that if $f$ is a constant function, then from~\eqref{eqn-BL-inf}, $\mathcal{I}_f(T)=0$ for every subset $T$ of $[n]$
	and so $t$-$\mathcal{I}(f)$. On the other hand, if $t$-$\mathcal{I}(f)=0$, then from Theorem~\ref{thm-at-most}, it follows that
	$t$-$\sym{inf}(f)=0$ and so from the second point of Theorem~\ref{thm-max-t-inf} we have that $f$ is a constant function.
\end{proof}

\begin{remark}\label{rem-AL}
	Upper bounds on $\mathcal{I}_f(T)$ for $T$ with bounded size have been proved in~\cite{AL93}. Since $\sym{inf}_f(T)\leq \mathcal{I}_f(T)$, it follows
	that these upper bounds also hold for $\sym{inf}_f(T)$.
\end{remark}

\section{Discussion \label{sec-discuss}}
We have introduced a new definition of influence of a set of variables on a Boolean function which is based on the auto-correlation function. Using the new definition,
we have proved a number of results. In this section, we highlight the new insights into Boolean functions that are obtained from the new results which follow
from the new definition.

As proved in Section~\ref{subsec-prob}, the quantity $I_f(T)$ defined in~\cite{fischer02,blais2009testing} is half the value of the 
influence (namely, $\sym{inf}_f(T)$) that we have defined. Some results for $I_f(T)$ have been obtained earlier. 
Remark~\ref{rem-blais} mentions the results which were previously obtained in~\cite{fischer02,blais2009testing}. The quantity $I_f(T)$ was used in~\cite{fischer02,blais2009testing}
as a tool for junta testing. The crucial result for such testing is Proposition~\ref{prop-junta}. We have provided a new and simpler proof of this proposition. 
Apart from Proposition~\ref{prop-junta} and the results mentioned in Remark~\ref{rem-blais}, all other results in Section~\ref{sec-inf-set} and its
various subsections appear for the first time in this paper. We highlight interesting aspects of some of the new results, particularly those aspects which arise due to the 
auto-correlation function based definition. 

Theorem~\ref{thm-path-expansion} connects total influence
to the path expansion property of a set of vertices $A$ of the hypercube. This result provides a geometric interpretation of the notion of influence which generalises
the well known connection of the total influence of a single variable to the edge expansion property of $A$. The geometric interpretation of total influence in terms
of path expansion is obtained through the connection of the auto-correlation function to path expansion and the new definition of influence using the auto-correlation
function. The Fourier/Walsh transform and the auto-correlation function are well studied tools in the theory of Boolean functions. In Proposition~\ref{prop-ac-path} and
Remark~\ref{rem-geom} we have explained the new geometric insight into these tools that our results provide. 

The notion of influence has been studied for a long time, but has been restricted mostly to issues in theoretical computer science. On the other hand, the notions
of bent functions and resilient functions have also been studied for a long time in the coding theory and cryptography literature. Our results provide
a previously unknown bridge between the notion of influence on the one hand, and the notions of bent and resilient functions on the other hand. The first point of 
Theorem~\ref{thm-max-t-inf}
provides a characterisation of resilient functions in terms of total influence. Theorem~\ref{thm-bent-inf} provides a characterisation of bent functions
in terms of influence. Theorem~\ref{thm-max-t-inf} is itself based on the characterisation of total influence in terms of Fourier/Walsh transform, while
the proof of Theorem~\ref{thm-bent-inf} uses the auto-correlation based definition of influence. These new results provide interesting new insights into the 
connection between aspects of Boolean functions studied in theoretical computer science and in coding theory and cryptography. 

Remark~\ref{rem-Tal} and the discussion following it mention the results on pseudo-influence which were previously obtained in~\cite{tal2017tight}. The other
results in Section~\ref{sec-PI} are new to this work. In particular, the inadequacy of pseudo-influence as a notion of influence is obtained as a consequence
of Theorem~\ref{thm-inf-basic}, and the characterisation of the conditions under which the total pseudo-influence achieves its minimum and maximum values
are given in Theorem~\ref{thm-max-t-PI}.

All results in Section~\ref{sec-BL} on the BL definition of influence are new to this paper. These results establish the basic properties
of this notion of influence. We provide a detailed comparison of the BL definition of influence and the auto-correlation function based definition of influence which highlight
why the BL definition is less satisfactory than the auto-correlation function based definition as a measure of influence. 

\section{Conclusion \label{sec-conclu} }
We introduced a definition of influence of a set of variables on a Boolean function using the auto-correlation function. The basic theory around the notion
of influence has been carefully developed and several well known results on the influence of a single variable have been generalised. New characterisations
of resilient and bent functions in terms of influence have been obtained. A previously introduced~\cite{fischer02,blais2009testing} measure of influence of a set of 
variables is shown to be half the value of the influence that we introduce. We also defined a notion of pseudo-influence, argued that it is not a satisfactory
measure of influence and showed that pseudo-influence is equal to a measure of influence previously defined in~\cite{tal2017tight}. 
Finally, we studied in details the definition of influence given by Ben-Or and Linial~\cite{ben1987collective} and brought out its relation to the auto-correlation based 
notion of influence.


\section*{Acknowledgement} We thank the reviewers of an earlier version for providing helpful comments.

\bibliographystyle{plain}
\bibliography{influence.bib}

\begin{thebibliography}{10}

\bibitem{AL93}
Mikl\'{o}s Ajtai and Nathal Linial.
\newblock The influence of large coalitions.
\newblock {\em Combinatorica}, 13(2):129--145, 1993.

\bibitem{ben1987collective}
Michael Ben{-}Or and Nathan Linial.
\newblock Collective coin flipping.
\newblock {\em Adv. Comput. Res.}, 5:91--115, 1989.

\bibitem{BS2021}
Aniruddha Biswas and Palash Sarkar.
\newblock Separation results for boolean function classes.
\newblock {\em Cryptography Commun.}, 13(3):451–458, may 2021.

\bibitem{blais2009testing}
Eric Blais.
\newblock Testing juntas nearly optimally.
\newblock In Michael Mitzenmacher, editor, {\em Proceedings of the 41st Annual
  {ACM} Symposium on Theory of Computing, {STOC} 2009, Bethesda, MD, USA, May
  31 - June 2, 2009}, pages 151--158. {ACM}, 2009.

\bibitem{DBLP:conf/eurocrypt/CanteautCCF00}
Anne Canteaut, Claude Carlet, Pascale Charpin, and Caroline Fontaine.
\newblock Propagation characteristics and correlation-immunity of highly
  nonlinear boolean functions.
\newblock In Bart Preneel, editor, {\em Advances in Cryptology - {EUROCRYPT}
  2000, International Conference on the Theory and Application of Cryptographic
  Techniques, Bruges, Belgium, May 14-18, 2000, Proceeding}, volume 1807 of
  {\em Lecture Notes in Computer Science}, pages 507--522. Springer, 2000.

\bibitem{carlet2021boolean}
Claude Carlet.
\newblock {\em Boolean Functions for Cryptography and Coding Theory}.
\newblock Cambridge University Press, January 2021.

\bibitem{fischer02}
Eldar Fischer, Guy Kindler, Dana Ron, Shmuel Safra, and Alex Samorodnitsky.
\newblock Testing juntas.
\newblock In {\em 43rd Symposium on Foundations of Computer Science {(FOCS}
  2002), 16-19 November 2002, Vancouver, BC, Canada, Proceedings}, pages
  103--112. {IEEE} Computer Society, 2002.

\bibitem{friedgut1996every}
Ehud Friedgut and Gil Kalai.
\newblock Every monotone graph property has a sharp threshold.
\newblock {\em Proceedings of the American mathematical Society},
  124(10):2993--3002, 1996.

\bibitem{DBLP:journals/iacr/GangopadhyayS14}
Sugata Gangopadhyay and Pantelimon St{\u{a}}nic{\u{a}}.
\newblock The {F}ourier entropy-influence conjecture holds for a log-density 1
  class of cryptographic boolean functions.
\newblock Cryptology ePrint Archive, Paper 2014/054, 2014.
\newblock \url{https://eprint.iacr.org/2014/054}.

\bibitem{Ka16}
Gil Kalai.
\newblock Boolean functions: Influence, threshold and noise.
\newblock In {\em European Congress of Mathematics (2016)}, pages 85--110.
  European Mathematical Society, 2018.

\bibitem{o2014analysis}
Ryan O'Donnell.
\newblock {\em Analysis of Boolean Functions}.
\newblock Cambridge University Press, 2014.

\bibitem{preneel1990propagation}
Bart Preneel, Werner~Van Leekwijck, Luc~Van Linden, Ren{\'{e}} Govaerts, and
  Joos Vandewalle.
\newblock Propagation characteristics of {B}oolean functions.
\newblock In {\em Advances in Cryptology - {EUROCRYPT} '90, Workshop on the
  Theory and Application of of Cryptographic Techniques}, volume 473, pages
  161--173. Springer, 1990.

\bibitem{DBLP:journals/jct/Rothaus76}
O.~S. Rothaus.
\newblock On ``bent'' functions.
\newblock {\em J. Comb. Theory, Ser. {A}}, 20(3):300--305, 1976.

\bibitem{DBLP:journals/tit/SarkarM04}
Palash Sarkar and Subhamoy Maitra.
\newblock Construction of nonlinear resilient boolean functions using ``small''
  affine functions.
\newblock {\em {IEEE} Trans. Inf. Theory}, 50(9):2185--2193, 2004.

\bibitem{DBLP:journals/tit/Siegenthaler84}
Thomas Siegenthaler.
\newblock Correlation-immunity of nonlinear combining functions for
  cryptographic applications.
\newblock {\em {IEEE} Trans. Inf. Theory}, 30(5):776--780, 1984.

\bibitem{tal2017tight}
Avishay Tal.
\newblock Tight bounds on the {F}ourier spectrum of {AC0}.
\newblock In Ryan O'Donnell, editor, {\em 32nd Computational Complexity
  Conference, {CCC} 2017, July 6-9, 2017, Riga, Latvia}, volume~79 of {\em
  LIPIcs}, pages 15:1--15:31. Schloss Dagstuhl - Leibniz-Zentrum f{\"{u}}r
  Informatik, 2017.

\bibitem{DBLP:journals/tit/XiaoM88}
Guo{-}Zhen Xiao and James~L. Massey.
\newblock A spectral characterization of correlation-immune combining
  functions.
\newblock {\em {IEEE} Trans. Inf. Theory}, 34(3):569--571, 1988.

\end{thebibliography}

\end{document}